\documentclass[a4paper,twocolumn,superscriptaddress,accepted=2017-12-06]{quantumarticle}
\pdfoutput=1

\usepackage{amsthm,amsmath,amssymb}
\usepackage{mathtools} 
\usepackage{bbm}
\usepackage{bm}

\usepackage[numbers,sort&compress]{natbib}

\usepackage[normalem]{ulem}

\newtheorem{theorem}{Theorem}
\newtheorem{proposition}[theorem]{Proposition}

\newtheorem{lemma}[theorem]{Lemma}
\newtheorem{definition}[theorem]{Definition}

\newcommand{\N}{\mathcal{N}}
\newcommand{\Ptn}{\mathcal{P}}
\newcommand{\A}{\mathcal{A}}
\newcommand{\K}{\mathcal{K}}

\usepackage[utf8]{inputenc}
\usepackage[T1]{fontenc}

\begin{document}

\title{Genuinely Multipartite Noncausality}

\date{December 7, 2017}

\author{Alastair A. Abbott}
\orcid{0000-0002-2759-633X}
\affiliation{University Grenoble Alpes, CNRS, Grenoble INP, Institut N\'eel, 38000 Grenoble, France}

\author{Julian Wechs}
\orcid{0000-0002-0395-6791}
\affiliation{University Grenoble Alpes, CNRS, Grenoble INP, Institut N\'eel, 38000 Grenoble, France}

\author{Fabio Costa}
\orcid{0000-0002-6547-6005}
\affiliation{Centre for Engineered Quantum Systems, School of Mathematics and Physics, The University of Queensland, St Lucia, QLD 4072, Australia}

\author{Cyril Branciard}
\orcid{0000-0001-9460-825X}
\affiliation{University Grenoble Alpes, CNRS, Grenoble INP, Institut N\'eel, 38000 Grenoble, France}

\maketitle

\begin{abstract}

The study of correlations with no definite causal order has revealed a rich structure emerging when more than two parties are involved. 
This motivates the consideration of multipartite ``noncausal'' correlations that cannot be realised even if noncausal resources are made available to a smaller number of parties. 
Here we formalise this notion: genuinely $N$-partite noncausal correlations are those that cannot be produced by grouping $N$ parties into two or more subsets, where a causal order between the subsets exists. 
We prove that such correlations can be characterised as lying outside a polytope, whose vertices correspond to deterministic strategies and whose facets define what we call ``2-causal'' inequalities. 
We show that genuinely multipartite noncausal correlations arise within the process matrix formalism, where quantum mechanics holds locally but no global causal structure is assumed, although for some inequalities no violation was found. 
We further introduce two refined definitions that allow one to quantify, in different ways, to what extent noncausal correlations correspond to a genuinely multipartite resource.

\end{abstract}

\section{Introduction}

Understanding the correlations between events, or between the parties that observe them, is a central objective in science. 
In order to provide an explanation for a given correlation, one typically refers to the notion of causality and embeds events (or parties) into a \emph{causal structure}, that defines a \emph{causal order} between them~\cite{Reichenbachbook,Pearlbook}. 
Correlations that can be explained in such a way, i.e., that can be established according to a definite causal order, are said to be \emph{causal}~\cite{brukner14}.

The study of causal correlations has gained a lot of interest recently as a result of the realisation that more general frameworks can actually be considered, where the causal assumptions are weakened and in which \emph{noncausal correlations} can be obtained~\cite{oreshkov12}. 
Investigations of causal versus noncausal correlations first focused on the simplest bipartite case~\cite{oreshkov12,branciard16}, and were soon extended to multipartite scenarios, where a much richer situation is found~\cite{baumeler13,baumeler14,oreshkov16,abbott16}---this opens, for instance, the possibility for causal correlations to be established following a dynamical causal order, where the causal order between events may depend on events occurring beforehand~\cite{hardy2005probability}.
When analysing noncausal correlations in a multipartite setting, however, a natural question arises: is the noncausality of these correlations a truly multipartite phenomenon, or can it be reduced to a simpler one, that involves fewer parties? 
The goal of this paper is precisely to address this question, and provide criteria to justify whether one really deals with \emph{genuinely multipartite noncausality} or not.

To make things more precise, let us start with the case of two parties, $A$ and $B$. 
Each party receives an input $x$, $y$, and returns an output $a$, $b$, respectively. 
The correlations shared by $A$ and $B$ are described by the conditional probability distribution $P(a,b|x,y)$. 
If the two parties' events (returning an output upon receiving an input) are embedded into a fixed causal structure, then one could have that $A$ causally precedes $B$---a situation that we shall denote by $A \prec B$, and where $B$'s output may depend on $A$'s input but not vice versa: $P(a|x,y) = P(a|x)$---or that $B$ causally precedes $A$---$B \prec A$, where $P(b|x,y) = P(b|y)$. 
(It can also be that the correlation is not due to a direct causal relation between $A$ and $B$, but to some latent common cause; such a situation is however still compatible with an explanation in terms of $A \prec B$ or $B \prec A$, and is therefore encompassed in the previous two cases.) 
A causal correlation is defined as one that is compatible with either $A \prec B$ or $B \prec A$, or with a convex mixture thereof, which would describe a situation where the party that comes first is selected probabilistically in each run of the experiment~\cite{oreshkov12,brukner14}.

Adding a third party $C$ with input $z$ and output $c$, and taking into account the possibility of a dynamical causal order, a tripartite causal correlation is defined as one that is compatible with one party acting first---which one it is may again be chosen probabilistically---and such that whatever happens with that first party, the reduced bipartite correlation shared by the other two parties, conditioned on the input and output of the first party, is causal (see Definition~\ref{def:fully_causal} below for a more formal definition, and its recursive generalisation to $N$ parties)~\cite{oreshkov16,abbott16}. 
In contrast, a \emph{noncausal} tripartite correlation $P(a,b,c|x,y,z)$ \emph{cannot} for instance be decomposed as
\begin{equation}\label{eq:tripartite_2causal}
	P(a,b,c|x,y,z) = P(a|x) \,P_{x,a}(b,c|y,z) 
\end{equation}
with bipartite correlations $P_{x,a}(b,c|y,z)$ that are causal for each $x,a$. 
Nevertheless, such a decomposition may still be possible for a tripartite noncausal correlation if one does not demand that (all) the bipartite correlations $P_{x,a}(b,c|y,z)$ are causal. 
Without this constraint, the correlation~\eqref{eq:tripartite_2causal} is thus compatible with the ``coarse-grained'' causal order $A \prec \{B,C\}$, if $B$ and $C$ are grouped together to define a new ``effective party'' and act ``as one''.
This illustrates that although a multipartite correlation may be noncausal, there might still exist some definite causal order between certain subsets of parties; the intuition that motivates our work is that such a correlation would therefore not display genuinely multipartite noncausality.

This paper is organised as follows. 
In Sec.~\ref{sec:genuine}, we introduce the notion of \emph{genuinely $N$-partite noncausal correlations} in opposition to what we call \emph{2-causal correlations}, which can be established whenever two separate groups of parties can be causally ordered; we furthermore show how such correlations can be characterised via so-called \emph{2-causal inequalities}. 
In Sec.~\ref{sec:lazy_scenario}, as an illustration we analyse in detail the simplest nontrivial tripartite scenario where these concepts make sense; we present explicit 2-causal inequalities for that scenario, investigate their violations in the process matrix framework of Ref.~\cite{oreshkov12}, and generalise some of them to $N$-partite inequalities. 
In Sec.~\ref{sec:refining}, we propose two possible generalisations of the notion of 2-causal correlations, which we call \emph{$M$-causal} and \emph{size-$S$-causal correlations}, respectively. 
This allows one to refine the analysis, and provides two different hierarchies of criteria that quantify the extent to which the noncausality of a correlation is a genuinely multipartite phenomenon.

\section{Genuinely $N$-partite noncausal correlations}
\label{sec:genuine}

The  general multipartite scenario that we consider in this paper, and the notations we use, are the same as in Ref.~\cite{abbott16}. 
A finite number $N \ge 1$ of parties $A_k$ each receive an input $x_k$ from some finite set (which can in principle be different for each party) and generate an output $a_k$ that also belongs to some finite set (and which may also differ for each input).  
The vectors of inputs and outputs are denoted by $\vec x = (x_1, \ldots, x_N)$ and $\vec a = (a_1, \ldots, a_N)$. 
The correlations between the $N$ parties are given by the conditional probability distribution $P(\vec a|\vec x)$.
For some (nonempty) subset $\K = \{k_1, \ldots, k_{|\K|}\}$ of $\N \coloneqq \{1, \ldots, N\}$, we denote by $\vec x_\K = (x_{k_1}, \ldots, x_{k_{|\K|}})$ and $\vec a_\K = (a_{k_1}, \ldots, a_{k_{|\K|}})$ the vectors of inputs and outputs of the parties in $\K$; with this notation, $\vec x_{\N \backslash \K}$ and $\vec a_{\N \backslash \K}$ (or simply $\vec x_{\N \backslash k}$ and $\vec a_{\N \backslash k}$ for a singleton $\K = \{k\}$) denote the vectors of inputs and outputs of all parties that are not in $\K$.
For simplicity we will identify the parties' names with their labels, so that $\N = \{1, \ldots, N\} \equiv \{A_1, \ldots, A_N\}$, and similarly for any subset $\K$.

\subsection{Definitions}

The assumption that the parties in such a scenario are embedded into a well-defined causal structure restricts the correlations that they can establish.
In Refs.~\cite{oreshkov16,abbott16}, the most general correlations that are compatible with a definite causal order between the parties were studied and characterised.
Such correlations include those compatible with causal orders that are probabilistic or dynamical---that is, the operations of parties in the past can determine the causal order of parties in the future.
These so-called \textit{causal correlations}---which, for clarity, we shall often call \textit{fully causal} here---can be defined iteratively in the following way:
\begin{definition}[(Fully) causal correlations] \label{def:fully_causal} $ $
	\begin{itemize}
		\item For $N=1$, any valid probability distribution $P(a_1|x_1)$ is (fully) causal;
		\item For $N \ge 2$, an $N$-partite correlation is (fully) causal if and only if it can be decomposed in the form
		\begin{equation}
			P(\vec a|\vec x) = \sum_{k \in \N} \ q_k \ P_k(a_k|x_k) \ P_{k,x_k,a_k}(\vec a_{\N \backslash k}|\vec x_{\N \backslash k}) 
		\end{equation}
		with $q_k \ge 0$ for each $k$, $\sum_k q_k = 1$, where (for each $k$) $P_k(a_k|x_k)$ is a single-party probability distribution and (for each $k, x_k, a_k$) $P_{k,x_k,a_k}(\vec a_{\N \backslash k}|\vec x_{\N\backslash k})$ is a (fully) causal $(N{-}1)$-partite correlation.
	\end{itemize}
\end{definition}

As the tripartite example in the introduction shows, there can be situations in which no overall causal order exists, but where there still is a (``coarse-grained'') causal order between certain subsets of parties, obtained by grouping certain parties together. 
The correlations that can be established in such situations are more general than causal correlations, but nevertheless restricted due to the existence of this partial causal ordering.
If we want to identify the idea of noncausality as a \emph{genuinely $N$-partite} phenomenon, we should, however, exclude such correlations, and characterise correlations for which no subset of parties can have a definite causal relation to any other subset. 
This idea was already suggested in Ref.~\cite{abbott16}; here we define the concept precisely.

Note that if several different nonempty subsets do have definite causal relations to each other, then clearly there will be two subsets having a definite causal relation between them---one can consider the subset that comes first and group the remaining subsets together into the complementary subset, which then comes second. 
We shall for now consider partitions of $\N$ into just two (nonempty) subsets $\K$ and $\N \backslash \K$, and we thus introduce the following definition:

\begin{definition}[$2$-causal correlations] \label{def:2_causal}
	An $N$-partite correlation (for $N \ge 2$) is said to be $2$-causal if and only if it can be decomposed in the form
	\begin{equation}\label{eqdef:2_causal}
		P(\vec a|\vec x) =\hspace{-2mm} \sum_{\emptyset \subsetneq \K \subsetneq \N} q_\K \, P_\K(\vec a_\K|\vec x_\K) \, P_{\K,\vec x_\K,\vec a_\K}(\vec a_{\N \backslash \K}|\vec x_{\N \backslash \K})
	\end{equation}
	where the sum runs over all nonempty strict subsets $\K$ of $\N$, with $q_\K \ge 0$ for each $\K$, $\sum_{\K} q_\K = 1$, and where (for each $\K$) $P_\K(\vec a_\K|\vec x_\K)$ is a valid probability distribution for the parties in $\K$ and (for each $\K, \vec x_\K,\vec a_\K$) $P_{\K,\vec x_\K,\vec a_\K}(\vec a_{\N \backslash \K}|\vec x_{\N \backslash \K})$ is a valid probability distribution for the remaining $N{-}|\K|$ parties.
\end{definition}

For ${N=2}$, the above definition reduces to the standard definition of bipartite causal correlations~\cite{oreshkov12}, which is equivalent to Definition~\ref{def:fully_causal} above. 
In the general multipartite case, it can be understood in the following way: each individual summand $P_\K(\vec a_\K|\vec x_\K) \, P_{\K,\vec x_\K,\vec a_\K}(\vec a_{\N \backslash \K}|\vec x_{\N \backslash \K})$ for each bipartition $\{ \K, \N \backslash \K \}$ describes correlations compatible with all the parties in $\K$ acting before all the parties in $\N \backslash \K$, since the choice of inputs for the parties in $\N \backslash \K$ does not affect the outputs for the parties in $\K$. 
The convex combination in Eq.~\eqref{eqdef:2_causal} then takes into account the possibility that the subset $\K$ acting first can be chosen randomly.%
\footnote{One can easily see that it is indeed sufficient to consider just one term per bipartition $\{ \K, \N \backslash \K \}$ in the sum~\eqref{eqdef:2_causal}. 
That is, for some given $\K$,  some correlations $P'(\vec a|\vec x) = P'_\K(\vec a_\K|\vec x_\K) \, P'_{\K,\vec x_\K,\vec a_\K}(\vec a_{\N \backslash \K}|\vec x_{\N \backslash \K})$ and $P''(\vec a|\vec x) = P''_\K(\vec a_\K|\vec x_\K) \, P''_{\K,\vec x_\K,\vec a_\K}(\vec a_{\N \backslash \K}|\vec x_{\N \backslash \K})$, and some weights $q', q'' \ge 0$ with $q' + q'' = 1$, the convex mixture $P(\vec a|\vec x) = q' P'(\vec a|\vec x) + q'' P''(\vec a|\vec x)$ is also of the same form $P(\vec a|\vec x) = P_\K(\vec a_\K|\vec x_\K) \, P_{\K,\vec x_\K,\vec a_\K}(\vec a_{\N \backslash \K}|\vec x_{\N \backslash \K})$ (with $P_\K(\vec a_\K|\vec x_\K) = q' P'_\K(\vec a_\K|\vec x_\K) + q'' P''_\K(\vec a_\K|\vec x_\K)$ and $P_{\K,\vec x_\K,\vec a_\K} = P(\vec a|\vec x) / P_\K(\vec a_\K|\vec x_\K)$). 
This already implies, in particular, that 2-causal correlations form a convex set.}

For correlations that are not 2-causal, we introduce the following terminology:
\begin{definition}[Genuinely $N$-partite noncausal correlations] \label{def:genuinely_N_partite_noncausal}
	An $N$-partite correlation that is not 2-causal is said to be genuinely $N$-partite noncausal.
\end{definition}

\noindent 
Thus, genuinely $N$-partite noncausal correlations are those for which it is impossible to find any definite causal relation between any two (complementary) subsets of parties, even when taking into consideration the possibility that the subset acting first may be chosen probabilistically.

\subsection{Characterisation of the set of 2-causal correlations as a convex polytope}

As shown in Ref.~\cite{branciard16} for the bipartite case and in Refs.~\cite{oreshkov16,abbott16} for the general $N$-partite case, any fully causal correlation can be written as a convex combination of deterministic fully causal correlations. 
As the number of such deterministic fully causal correlations is finite (for finite alphabets of inputs and outputs), they correspond to the extremal points of a convex polytope---the \emph{(fully) causal polytope}. 
The facets of this polytope are given by linear inequalities, which define so-called \emph{(fully) causal inequalities}.

As it turns out, the set of $2$-causal correlations can be characterised as a convex polytope in the same way:

\begin{theorem}
\label{thm:2_causal_polytope}
The set of 2-causal correlations forms a convex polytope, whose (finitely many) extremal points correspond to deterministic 2-causal correlations.
\end{theorem}

\begin{proof}
For a given nonempty strict subset $\K$ of $\N$, $P_\K(\vec a_\K|\vec x_\K) \, P_{\K,\vec x_\K,\vec a_\K}(\vec a_{\N \backslash \K}|\vec x_{\N \backslash \K})$ defines an ``effectively bipartite'' correlation, that is, a bipartite correlation between an effective party $\K$ with input $\vec x_\K$ and output $\vec a_\K$ and an effective party $\N \backslash \K$ with input $\vec x_{\N \backslash \K}$ and output $\vec a_{\N \backslash \K}$, which are formed by grouping together all parties in the respective subsets. 
That effectively bipartite correlation is compatible with the causal order%
\footnote{The notation $\K_1 \prec \K_2$ (or simply $A_{k_1} \prec A_{k_2}$ for singletons $\K_j = \{A_{k_j}\}$), already used in the introduction, formally means that the correlation under consideration satisfies $P(\vec a_{\K_1}|\vec x) = P(\vec a_{\K_1}|\vec x_{\N \backslash \K_2})$.  
It will also be extended to more subsets, with $\K_1 \prec \K_2 \prec \cdots \prec \K_m$ meaning that $P(\vec a_{\K_1\cup\cdots\cup\K_j}|\vec x) = P(\vec a_{\K_1\cup\cdots\cup\K_j}|\vec x_{\N \backslash (\K_{j+1}\cup\cdots\cup\K_m)})$ for all $j = 1, \ldots, m-1$.} 
$\K \prec \N \backslash \K$. 
As mentioned above, the set of such correlations forms a convex polytope whose extremal points are deterministic, effectively bipartite causal correlations~\cite{branciard16}---which, according to Definition~\ref{def:2_causal}, define deterministic 2-causal $N$-partite correlations.

Eq.~\eqref{def:2_causal} then implies that the set of 2-causal correlations is the convex hull of all such polytopes for each nonempty strict subset $\K$ of $\N$; it is thus itself a convex polytope, whose extremal points are indeed deterministic 2-causal correlations.
\end{proof}

\begin{figure}
	\begin{center}
	\includegraphics[width=.95\columnwidth]{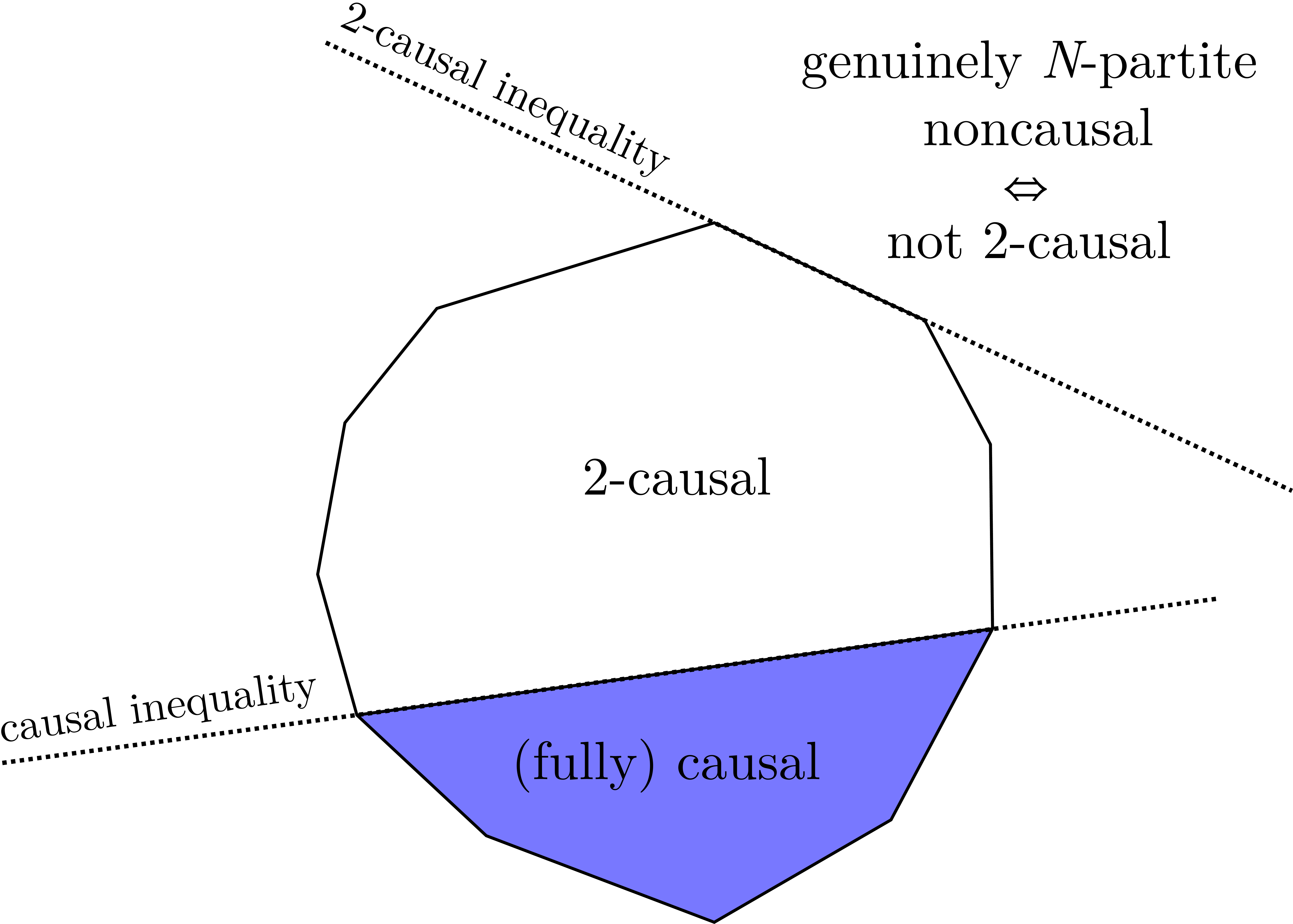}
	\end{center}
	\caption{Sketch of the fully causal and 2-causal polytopes (the shading should be interpreted as indicating that the latter contains the former). The vertices of the polytopes correspond to deterministic fully causal and 2-causal correlations, and their facets correspond to causal and 2-causal inequalities, respectively. Correlations that are outside of the fully causal polytope are simply noncausal; correlations that are outside of the 2-causal polytope are genuinely $N$-partite noncausal.}
	\label{fig:2polytopes}
\end{figure}

As any fully causal correlation is 2-causal, but not vice versa, the fully causal polytope is a strict subset of what we shall call the \emph{2-causal polytope} (see Fig.~\ref{fig:2polytopes}). 
Every vertex of the 2-causal polytope corresponds to a deterministic function $\vec \alpha$ that assigns a list of outputs $\vec a = \vec \alpha(\vec x)$ to the list of inputs $\vec x$, such that the corresponding probability distribution $P_{\vec \alpha}^\text{det}(\vec a|\vec x) = \delta_{\vec a, \vec \alpha(\vec x)}$ is 2-causal, and thus satisfies Eq.~\eqref{eqdef:2_causal}.
Since $P_{\vec \alpha}^\text{det}(\vec a|\vec x)$ can only take values $0$ or $1$, there is only one term in the sum in Eq.~\eqref{eqdef:2_causal}, and it can be written such that there is a single (nonempty) strict subset $\K$ that acts first. 
That is, $\vec \alpha$ is such that the outputs $\vec a_\K$ of the parties in $\K$ are determined exclusively by their inputs $\vec x_\K$, while the outputs $\vec a_{\N \backslash \K}$ of the remaining parties are determined by all inputs $\vec x$. 
The facets of the 2-causal polytope are linear inequalities that are satisfied by all 2-causal correlations; we shall call these \textit{2-causal inequalities} (see Fig.~\ref{fig:2polytopes}).

\section{Analysis of the tripartite ``lazy scenario''}
\label{sec:lazy_scenario}

In this section we analyse in detail, as an illustration, the polytope of 2-causal correlations for the simplest nontrivial scenario with more than two parties.
In Ref.~\cite{abbott16} it was shown that this scenario is the so-called tripartite ``lazy scenario'', in which each party $A_k$ receives a binary input $x_k$, has a single constant output for one of the inputs, and a binary output for the other.
By convention we consider that for each $k$, on input $x_k=0$ the output is always $a_k=0$, while for $x_k=1$ we take $a_k\in\{0,1\}$.
The set of fully causal correlations was completely characterised for this scenario in Ref.~\cite{abbott16}, which will furthermore permit us to compare the noncausal and genuinely tripartite noncausal correlations in this concrete example.

As is standard (and as we did in the introduction), we will denote here the three parties by $A$, $B$, $C$, their inputs $x$, $y$, $z$, and their outputs $a$, $b$ and $c$.
Furthermore, we will denote the complete tripartite probability distribution by $P_{ABC}$ [i.e., $P_{ABC}(abc|xyz)\coloneqq P(abc|xyz)$] and the marginal distributions for the indicated parties by $P_{AB}$, $P_{A}$, etc. [e.g., $P_{AB}(ab|xyz)=\sum_c P_{ABC}(abc|xyz)$].

\subsection{Characterisation of the polytope of 2-causal correlations}

\subsubsection{Complete characterisation}

We characterise the polytope of 2-causal correlations in much the same way as the polytope of fully causal correlations was characterised in Ref.~\cite{abbott16}, where we refer the reader for a more in-depth presentation.
Specifically, the vertices of the polytope are found by enumerating all deterministic 2-causal probability distributions $P_{ABC}$, i.e., those which admit a decomposition of the form~\eqref{eqdef:2_causal} with (because they are deterministic) a single term in the sum (corresponding to a single group of parties acting first).
One finds that there are $1\,520$ such distributions, and thus vertices.

In order to determine the facets of the polytope, which in turn correspond to tight 2-causal inequalities, a parametrisation of the 19-dimensional polytope must be fixed and the convex hull problem solved.
We use the same parametrisation as in Ref.~\cite{abbott16}, and again use \textsc{cdd}~\cite{cdd} to compute the facets of the polytope.
We find that the polytope has $21\,154$ facets, each corresponding to a 2-causal inequality, the violation of which would certify genuinely tripartite noncausality.
Many inequalities, however, can be obtained from others by either relabelling outputs or permuting parties, and as a result it is natural to group the inequalities into equivalence classes, or ``families'', of inequalities.
Taking this into account, we find that there are 476 families of facet-inducing 2-causal inequalities, 3 of which are trivial, as they simply correspond to positivity constraints on the probabilities (and are thus satisfied by any valid probability distribution).
While the 2-causal inequalities all detect genuinely $N$-partite noncausality, it is interesting to note that all except 22 of them can be saturated by fully causal correlations (and all but 37 even by correlations compatible with a fixed causal order).

We provide the complete list of these inequalities, organised by their symmetries and the types of distribution required to saturate them, in the Supplementary Material~\cite{SM}, and will analyse in more detail a few particularly interesting examples in what follows.
First, however, it is interesting to note that only 2 of the 473 nontrivial facets are also facets of the (fully) causal polytope for this scenario (one of which is Eq.~\eqref{eq:ineq3} analysed below), and hence the vast majority of facet-inducing inequalities of the causal polytope do not single out genuinely tripartite noncausal correlations.
Moreover, none of the 2-causal inequalities we obtain here differ from facet-inducing fully causal inequalities only in their bound, and, except for the aforementioned cases, our 2-causal inequalities thus represent novel inequalities.

\subsubsection{Three interesting inequalities}

Of the nontrivial 2-causal inequalities, those that display certain symmetries between the parties are particularly interesting since they tend to have comparatively simple forms and often permit natural interpretations (e.g., as causal games~\cite{oreshkov12,branciard16}).

For example, three nontrivial families of 2-causal inequalities have forms (i.e., certain versions of the inequality within the corresponding equivalence class) that are completely symmetric under permutations of the parties.
One of these is the inequality
\begin{align}\label{eq:ineq1}
 I_1 &= \big[ P_A(1|100) + P_B(1|010) + P_C(1|001) \big] \notag \\
&+ \big[ P_{AB}(11|110) + P_{BC}(11|011) + P_{AC}(11|101) \big] \notag \\
&- P_{ABC}(111|111)  \ge 0,
\end{align}
which can be naturally expressed as a causal game.  Indeed, it can be rewritten as
\begin{equation}\label{eq:ineq1game}
P\big( \tilde{a}\tilde{b}\tilde{c} = x y z \big) \ \le \ 3/4 \,,
\end{equation}
where $\tilde{a}=1$ if $x=0$, $\tilde{a}=a$ if $x=1$ 
(i.e., $\tilde{a}=xa\oplus x \oplus 1$, where $\oplus$ denotes addition modulo 2), 
and similarly for $\tilde b$ and $\tilde c$, and where it is implicitly assumed that all inputs occur with the same probability.
This can be interpreted as a game in which the goal is to collaborate such that the product of the nontrivial outputs (i.e., those corresponding to an input 1) is equal to the product of the inputs, and where the former product is taken to be 1 if all inputs are 0 and there are therefore no nontrivial outputs (in which case the game will always be lost).
The probability of success for this game can be no greater than $3/4$ if the parties share a 2-causal correlation.
This bound can easily be saturated by a deterministic, even fully causal, distribution: if every party always outputs 0 then the parties will win the game in all cases, except when the inputs are all 0 or all 1.

Another party-permutation-symmetric 2-causal inequality is the following:
\begin{align}\label{eq:ineq2}
I_2 = 1 + 2 \big[ P_A(1|100) + P_B(1|010) + P_C(1|001) \big] \notag \\
- \big[ P_{AB}(11|110) + P_{BC}(11|011) + P_{AC}(11|101) \big] & \ge 0,
\end{align}
whose interpretation can be made clearer by rewriting it as
\begin{align}\label{eq:3wayLGYNI}
	P_A(1|100) + P_B(1|010) - P_{AB}(11|110) & \notag\\
	+\, P_B(1|010) + P_C(1|001) - P_{BC}(11|011) & \notag\\
	+\, P_A(1|100) + P_C(1|001) - P_{AC}(11|101) & \ge -1.
\end{align}
The left-hand side of this inequality is simply the sum of three terms corresponding to conditional ``lazy guess your neighbour's input'' (\mbox{LGYNI}) inequalities~\cite{abbott16}, one for each pair of parties (conditioned on the remaining party having input $0$), while the negative bound on the right-hand side accounts for the fact that any pair of parties that are grouped together in a bipartition may maximally violate the \mbox{LGYNI} inequality between them (and thus reach the minimum algebraic bound $-1$).
This inequality can be interpreted as a ``scored game'' (as opposed to a ``win-or-lose game'') in which each pair of parties scores one point if they win their respective bipartite \mbox{LGYNI} game and the third party's input is 0, and where the goal of the game is to maximise the total score, given by the sum of all three pairs' individual scores. 
The best average score (when the inputs are uniformly distributed) for a 2-causal correlation is $5/4$, corresponding to the 2-causal bounds of $0$ in Eq.~\eqref{eq:ineq2} and $-1$ in Eq.~\eqref{eq:3wayLGYNI}.%
\footnote{The bound of these inequalities, and the best average score of the corresponding game, can be reached by a 2-causal strategy in which one party, say $A$, has a fixed causal order with respect to the other two parties grouped together, who share a correlation maximally violating the corresponding \mbox{LGYNI} inequality. For example, the distribution $P(abc|xyz)=\delta_{a,0} \, \delta_{b,yz} \, \delta_{c,yz}$, where $\delta$ is the Kronecker delta function, is compatible with the order $A \prec \{B,C\}$ (or with $\{B,C\} \prec A$) and saturates Eqs.~\eqref{eq:ineq2} and~\eqref{eq:3wayLGYNI}.} 
It is also clear from the form of Eq.~\eqref{eq:3wayLGYNI} that for fully causal correlations the left-hand side is lower-bounded by $0$. 
This inequality is thus amongst the 22 facet-inducing 2-causal inequalities that cannot be saturated by fully causal distributions.

In addition to the inequalities that are symmetric under any permutation of the parties, there are four further nontrivial families containing 2-causal inequalities which are symmetric under cyclic exchanges of parties.
One interesting such example is the following:
\begin{align}
I_3 &= 2 + [ P_A(1|100) + P_B(1|010) + P_C(1|001) \big]  \notag \\
&- \big[ P_A(1|101) + P_B(1|110) + P_C(1|011) \big]  \ge 0. \label{eq:ineq3}
\end{align}
This inequality can again be interpreted as a causal game in the form (where we again implicitly assume a uniform distribution of inputs for all parties)
\begin{align}
& P\big(x(y \oplus 1)(a \oplus z)=y(z \oplus 1)(b \oplus x) \notag \\
& \hspace{20mm} =z(x \oplus 1)(c \oplus y)=0\big) \ \le \ 7/8 \,, \label{eq:ineq3game}
\end{align}
where the goal of the game is for each party, whenever they receive the input 1 and their right-hand neighbour has the input 0, to output the input of their left-hand neighbour (with $C$ being considered, in a circular manner, to be to the left of $A$).%
\footnote{The bound of $7/8$ on the probability of success can, for instance, be reached by the fully causal (and hence 2-causal) distribution $P(abc|xyz)=\delta_{a,x} \, \delta_{b,xy} \, \delta_{c,yz}$, compatible with the order $A\prec B \prec C$, which wins the game in all cases except when $(x,y,z)=(1,0,0)$.} 
This inequality is of additional interest as it is one of the two nontrivial inequalities which is also a facet of the standard causal polytope for this scenario.
(The second such inequality, which lacks the symmetry of this one, is presented in the Supplementary Material~\cite{SM}.)

\subsection{Violations of 2-causal inequalities by process matrix correlations}

One of the major sources of interest in causal inequalities has been the potential to violate them in more general frameworks, in which causal restrictions are weakened.
There has been a particular interest in one such model, the \emph{process matrix formalism}, in which quantum mechanics is considered to hold locally for each party, but no global causal order between the parties is assumed~\cite{oreshkov12}.
In this framework, the (possibly noncausal) interactions between the parties are described by a process matrix $W$, which, along with a description of the operations performed by the parties, allows the correlations $P(\vec{a}|\vec{x})$ to be calculated.

It is well-known that process matrix correlations can violate causal inequalities~\cite{abbott16,baumeler13,baumeler14,branciard16,oreshkov12}, although the physical realisability of such processes remains an open question~\cite{araujo17,feix16}.
In Ref.~\cite{abbott16} it was shown that all the nontrivial fully causal inequalities for the tripartite lazy scenario can be violated by process matrices.
However, for most inequalities violation was found to be possible using process matrices $W^{\{A,B\}\prec C}$ that are compatible with $C$ acting last, which means the correlations they produced were necessarily 2-causal.
It is therefore interesting to see whether process matrices are capable of violating 2-causal inequalities in general, and thus of exhibiting genuinely $N$-partite noncausality.
We will not present the process matrix formalism here, and instead simply summarise our findings;
we refer the reader to Refs.~\cite{araujo15,oreshkov12} for further details on the technical formalism.

Following the same approach as in Refs.~\cite{branciard16,abbott16} we looked for violations of the 2-causal inequalities.
Specifically, we focused on two-dimensional (qubit) systems and applied the same ``see-saw'' algorithm to iteratively perform semidefinite convex optimisation over the process matrix and the instruments defining the operations of the parties.

As a result, we were able to find process matrices violating all but 2 of the 473 nontrivial families of tight 2-causal inequalities (including Eqs.~\eqref{eq:ineq1} and~\eqref{eq:ineq3} above) using qubits, and in all cases where a violation was found, the best violation was given by the same instruments that provided similar results in Ref.~\cite{abbott16}.
We similarly found that 284 families of these 2-causal inequalities (including Eq.~\eqref{eq:ineq3}) could be violated by completely classical process matrices,\footnote{Incidentally, exactly the same number of families of fully causal inequalities were found to be violable with classical process matrices in Ref.~\cite{abbott16}. It remains unclear whether this is merely a coincidence or the result of a deeper connection.} a phenomenon that is not present in the bipartite scenario where classical processes are necessarily causal~\cite{oreshkov12}.

While the violation of 2-causal inequalities is again rather ubiquitous, the existence of two inequalities for which we found no violation is curious.
One of these inequalities is precisely Eq.~\eqref{eq:ineq2}, and its decomposition in Eq.~\eqref{eq:3wayLGYNI} into three \mbox{LGYNI} inequalities helps provide an explanation.
In particular, the seemingly best possible violation of a (conditional) \mbox{LGYNI} inequality using qubits is approximately $0.2776$~\cite{abbott16,branciard16}, whereas it is clear that a process matrix violating Eq.~\eqref{eq:3wayLGYNI} must necessarily violate a conditional \mbox{LGYNI} inequality between one pair of parties by at least $1/3$.
Moreover, in Ref.~\cite{branciard16} it was reported that no better violation was found using three- or four-dimensional systems, indicating that Eq.~\eqref{eq:3wayLGYNI} can similarly not be violated by such systems.
It nonetheless remains unproven whether such a violation is indeed impossible, and the convex optimisation problem for three parties quickly becomes intractable for higher dimensional systems, making further numerical investigation difficult.
The second inequality for which no violation was found can similarly be expressed as a sum of three different forms (i.e., relabellings) of a conditional \mbox{LGYNI} inequality, and a similar argument thus explains why no violation was found.
Recall that, as they can be expressed as a sum of three conditional \mbox{LGYNI} inequalities with a negative 2-causal bound, these two 2-causal inequalities cannot be saturated by fully causal distributions; it is interesting that the remaining inequalities that require noncausal but 2-causal distributions to saturate can nonetheless be violated by process matrix correlations.

\subsection{Generalised 2-causal inequalities for $N$ parties}

Although it quickly becomes intractable to completely characterise the 2-causal polytope for more complicated scenarios with more parties, inputs and/or outputs, as is also the case for fully causal correlations, it is nonetheless possible to generalise some of the 2-causal inequalities into inequalities that are valid for any number of parties $N$.

The inequality~\eqref{eq:ineq1}, for example, can naturally be generalised to give a 2-causal inequality valid for all $N \ge 2$.%
\footnote{We continue to focus on the lazy scenario defined earlier for concreteness, but we note that the proofs of the generalised inequalities~\eqref{eq:genIneq1} and~\eqref{eq:genIneq2} in fact hold in any nontrivial scenario, of which the lazy one is the simplest example. The bounds for the corresponding causal games and whether or not the inequalities define facets will, however, generally depend on the scenario considered.\label{footnote:genScenarios}}
Specifically, one obtains
\begin{align}\label{eq:genIneq1}
	J_1(N)=& \hspace{-2mm} \sum_{\emptyset \subsetneq \K \subsetneq \N}\hspace{-2mm} P\big(\vec{a}_{\K}=\vec{1}\,|\,\vec{x}_{\K}=\vec{1},\vec{x}_{\N\backslash \K}=\vec{0}\big) \qquad \notag \\[-2mm]
	& \hspace{22mm} - P\big(\vec{a}=\vec{1}\,|\,\vec{x}=\vec{1}\big)  \ge 0
\end{align}
where $\vec{1}=(1,\dots,1)$ and $\vec{0}=(0,\dots,0)$,
which can be written analogously to Eq.~\eqref{eq:ineq1game} as a game (again implicitly defined with uniform inputs) of the form
\begin{equation}
	P(\Pi_k \tilde{a}_k = \Pi_k x_k)\le 1-2^{-N+1}.
\end{equation}
We leave the proof of this inequality and its 2-causal bound to Appendix~\ref{appndx:genIneq}.
It is interesting to ask if this inequality is tight (i.e., facet inducing) for all $N$. For $N=2$ it reduces to the \mbox{LGYNI} inequality which is indeed tight, and for $N=3$ it was also found to be a facet.
By explicitly enumerating the vertices of the 2-causal polytope for $N=4$ (of which there are $136\,818\,592$) we were able to verify that $J_1(4)\ge 0$ is indeed also a facet, and we conjecture that this is true for all $N$.
Note that, as for the tripartite case it is trivial to saturate the inequality for all $N$ by considering the (fully causal) strategy where each party always outputs 0.

It is also possible to generalise inequality~\eqref{eq:3wayLGYNI} to $N$ parties---which will prove more interesting later---by considering a scored game in which every pair of parties gets one point if they win their respective bipartite \mbox{LGYNI} game and all other parties' inputs are 0, and the goal of the game is to maximise the total score of all pairs.
If two parties belong to the same subset in a bipartition, then they can win their respective \mbox{LGYNI} game perfectly, whereas they are limited by the causal bound $0$ if they belong to two different groups.
The 2-causal bound on the inequality is thus given by the maximum number of pairs of parties that belong to a common subset over all bipartitions, times the maximal violation of the bipartite \mbox{LGYNI} inequality.
Specifically, we obtain the 2-causal inequality
\begin{equation}\label{eq:genIneq2}
	J_2(N)=\hspace{-2mm}\sum_{\{i,j\}\subset \N} L_N(i,j)\ge -\binom{N-1}{2}
\end{equation}
where  $\binom{n}{2}=n(n{-}1)/2$ is a binomial coefficient and
\begin{align}
	 \hspace{-1mm} L_N(i,j) = \, & P\big(a_i=1|x_ix_j=10,\vec{x}_{\N\backslash \{i,j\}}=\vec{0}\big)\notag\\ 
		& \!\! + P\big(a_j=1|x_ix_j=01,\vec{x}_{\N\backslash \{i,j\}}=\vec{0}\big)\notag\\ \label{bilgynies}
		& \!\! - P\big(a_ia_j=11|x_ix_j=11,\vec{x}_{\N\backslash \{i,j\}}=\vec{0}\big).
\end{align}
Each term $L_N(i,j)$ defines a bipartite conditional \mbox{LGYNI} inequality with the causal bound $L_N(i,j)\ge 0$, and the minimum algebraic bound (i.e.\ the maximal violation) $-1$. The minimum algebraic bound of $J_2(N)$ is thus $-\binom{N}{2}$.
The validity of inequality~\eqref{eq:genIneq2} for 2-causal correlations (which corresponds to a maximal average score of $(2N{-}1)(N{-}1)/2^N$---compared to the maximal algebraic value of $2N(N{-}1)/2^N$---for the corresponding game with uniform inputs) is again formally proved in Appendix~\ref{appndx:genIneq}.

We note that in contrast to Eq.~\eqref{eq:genIneq1}, $J_2(4)\ge -3$ is not a facet of the 4-partite 2-causal polytope, and thus the inequality is not tight in general.
Inequality~\eqref{eq:genIneq2} can nonetheless be saturated by 2-causal correlations for any $N$.
For example, consider $\K=\{1,\dots,N-1\}$ and take the distribution
\begin{equation}\label{eq:genIneqViolatingDist}
	P(\vec{a}|\vec{x})= \delta_{\vec{a}_\K,f(\vec{x}_\K)} \, \delta_{a_N,0}
\end{equation}
with $f(\vec{x}_\K) = \vec{x}_\K$ if $\vec{x}_\K$ contains exactly two inputs 1, and $f(\vec{x}_\K) = \vec{0}$ otherwise.
$P(\vec{a}|\vec{x})$ is clearly 2-causal since it is compatible with the causal order $\K \prec \N \backslash \K$ (indeed, also with $\N \backslash \K \prec \K$).
One can then easily verify that $P(\vec{a}|\vec{x})$ saturates~\eqref{eq:genIneq2}, since all $\binom{N-1}{2}$ pairs of parties in $\K$ can win their respective conditional \mbox{LGYNI} game perfectly, and therefore contribute with a term of $-1$ to the sum in Eq.~\eqref{eq:genIneq2}.

\section{Refining the definition of genuinely multipartite noncausal correlations}
\label{sec:refining}

So far we only discussed correlations that can or cannot arise given a definite causal order between two subsets of parties. 
It makes sense to consider more refined definitions that discriminate, among noncausal correlations, to what extent and in which way they represent a genuinely multipartite resource. 
The idea will again be to see if a given correlation can be established by letting certain groups of parties act ``as one'', while retaining a definite causal order between different groups. 
The number and size of the groups for which this is possible can be used to give two distinct characterisations of how genuinely multipartite the observed noncausality is.

\subsection{$\Ptn$-causal correlations}
\label{subsec:Pcausal}

We first want to characterise the correlations that can be realised when a definite causal order exists between certain groups of parties, while no constraint is imposed on the correlations within each group.

\begin{figure}
	\begin{center}
		 \includegraphics[width=.95\columnwidth]{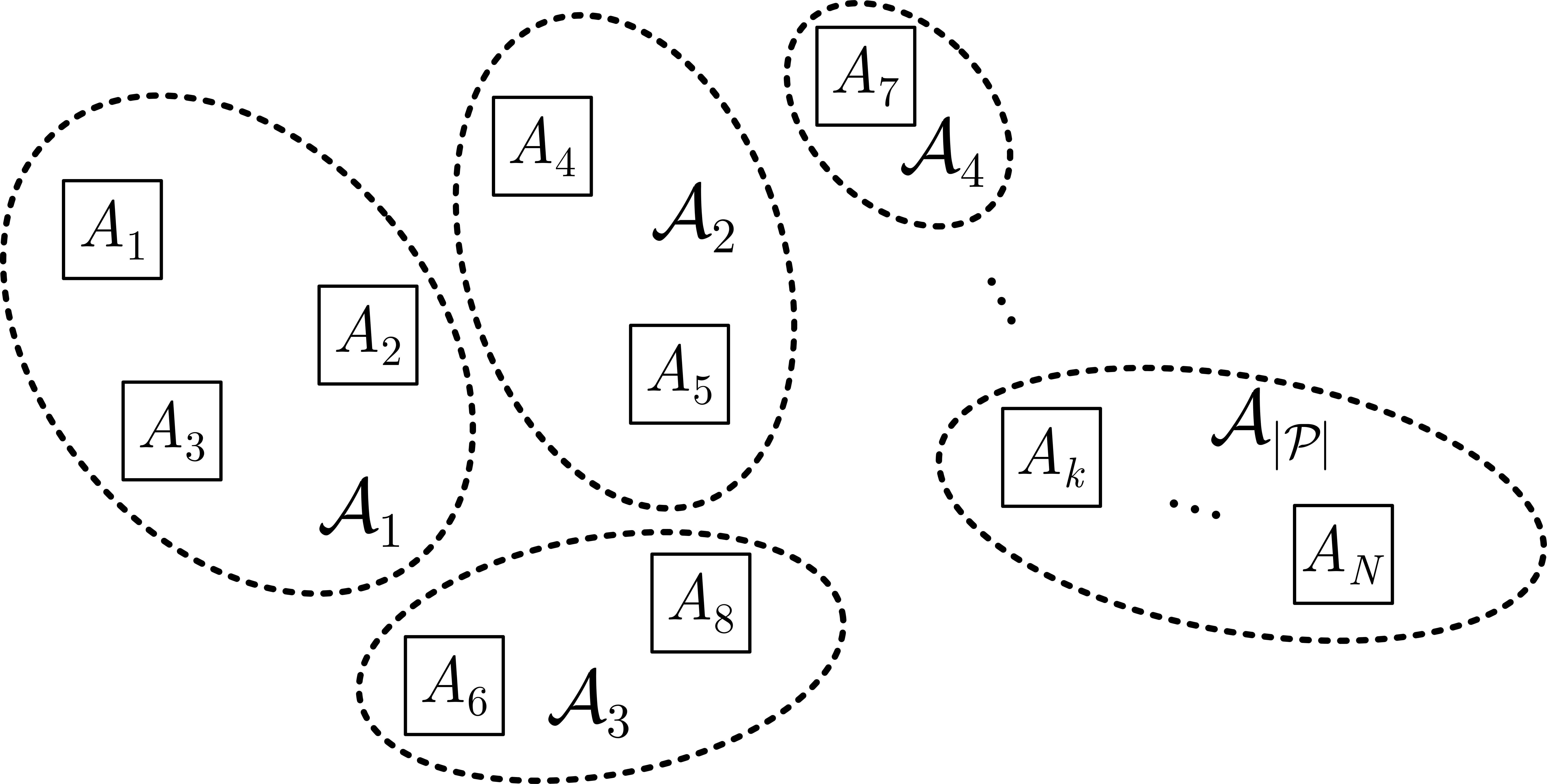}
	\end{center}
	\caption{The $N$ parties $A_k$ are grouped into $|\Ptn|$ disjoint subsets $\A_\ell$, according to the partition $\Ptn = \{ \A_1, \ldots, \A_{|\Ptn|} \}$. The parties in a subset $\A_\ell$ act ``as one'', thus defining an ``effective party''. To define $\Ptn$-causal correlations, the $N$-partite correlation $P(a_1,\ldots,a_N|x_1,\ldots,x_N)$ is considered as an effectively $|\Ptn|$-partite correlation $P({\vec a}_{\A_1},\ldots,{\vec a}_{\A_{|\Ptn|}}|{\vec x}_{\A_1},\ldots,{\vec x}_{\A_{|\Ptn|}})$.}
	\label{fig:partition}
\end{figure}

Let us consider for this purpose a partition $\Ptn = \{ \A_1, \ldots, \A_{|\Ptn|} \}$ of $\N$---i.e., a set of $|\Ptn|$ nonempty disjoint subsets $\A_\ell$ of $\N$, such that $\cup_\ell \A_\ell = \N$, see Fig.~\ref{fig:partition}. 
Note that if $\Ptn$ contains at least two subsets, then for a given subset $\A_\ell\subset \N$, $\Ptn \backslash \{\A_\ell\}$ also represents a partition of $\N \backslash \A_\ell$.
Let us then introduce the following definition:

\begin{definition}[$\Ptn$-causal correlations] \label{def:P_causal}
For a given partition $\Ptn$ of $\N$, an $N$-partite correlation $P$ is said to be $\Ptn$-causal if and only if $P$ is causal when considered as an effective $|\Ptn|$-partite correlation, where each subset in $\Ptn$ defines an effective party.

More precisely, analogously to Definition~\ref{def:fully_causal}:
\begin{itemize}
\item For $|\Ptn| = 1$, any $N$-partite correlation $P$ is $\Ptn$-causal;
\item For $|\Ptn| \ge 2$, an $N$-partite correlation $P$ is $\Ptn$-causal if and only if it can be decomposed in the form
	\begin{align}
		P(\vec a|\vec x) =& \sum_{{\A_\ell} \in \Ptn} q_{\A_\ell} \, P_{\A_\ell}(\vec a_{\A_\ell}|\vec x_{\A_\ell}) \hspace{30mm} \notag \\
	&\hspace{2mm}\times P_{\A_\ell,\vec x_{\A_\ell},\vec a_{\A_\ell}}(\vec a_{\N \backslash {\A_\ell}}|\vec x_{\N \backslash {\A_\ell}}) \quad \label{eqdef:P_causal}
	\end{align}
	with $q_{\A_\ell} \ge 0$ for each $\A_\ell$, $\sum_{\A_\ell} q_{\A_\ell} = 1$, where (for each $\A_\ell$) $P_{\A_\ell}(\vec a_{\A_\ell}|\vec x_{\A_\ell})$ is a valid probability distribution for the parties in $\A_\ell$ and (for each $\A_\ell,\vec x_{\A_\ell},\vec a_{\A_\ell}$) $P_{\A_\ell,\vec x_{\A_\ell},\vec a_{\A_\ell}}(\vec a_{\N \backslash {\A_\ell}}|\vec x_{\N \backslash {\A_\ell}})$ is a $(\Ptn \backslash \{\A_\ell\})$-causal correlation for the remaining $N-|\A_\ell|$ parties.
\end{itemize}
\end{definition}

In the extreme case of a single-set partition $\Ptn = \{\N\}$ ($|\Ptn| = 1$), any correlation is by definition trivially $\Ptn$-causal; 
at the other extreme, for a partition of $\N$ into $N$ singletons ($|\Ptn| = N$), the definition  of $\Ptn$-causal correlations above is equivalent to that of fully causal correlations, Definition~\ref{def:fully_causal}~\cite{oreshkov16,abbott16}.
Between these two extreme cases, a $\Ptn$-causal correlation identifies the situation where, with some probability, all parties within one group act before all other parties; 
conditioned on their inputs and outputs, another group acts second (before all remaining parties) with some probability; and so on. 
We emphasise that no constraint is imposed on the correlations that can be generated within each group, since we allow them to share the most general resource conceivable---in particular, there might be no definite causal order between the parties inside a group.

Since the definition of $\Ptn$-causal correlations above matches that of causal correlations for the $|\Ptn|$ effective parties defined by $\Ptn$, all basic properties of causal correlations (see Ref.~\cite{abbott16}) generalise straightforwardly to $\Ptn$-causal correlations. 
Note in particular that the definition captures the idea of dynamical causal order, where the causal order between certain subsets of parties in $\Ptn$ may depend on the inputs and outputs of other subsets of parties that acted before them. 
The following result also follows directly from what is known about causal correlations~\cite{oreshkov16,abbott16}:

\begin{theorem}
\label{thm:P_causal_polytope}
For any given $\Ptn$, the set of $\Ptn$-causal correlations forms a convex polytope, whose (finitely many) extremal points correspond to deterministic $\Ptn$-causal correlations.
\end{theorem}

We shall call this polytope the \emph{$\Ptn$-causal polytope}; its facets define \emph{$\Ptn$-causal inequalities}.
Theorem~\ref{thm:P_causal_polytope} implies that any $\Ptn$-causal correlation can be obtained as a probabilistic mixture of deterministic $\Ptn$-causal correlations. 
It is useful to note that, similarly to Ref.~\cite{abbott16}, deterministic $\Ptn$-causal correlations can be interpreted in the following way: a set $\A_{t_1}$ of parties acts with certainty before all others, with their outputs being a deterministic function of all inputs in that set but independent of the inputs of any other parties, $\vec a_{{\A_{t_1}}} = \vec \alpha_{{\A_{t_1}}}(\vec x_{ {\A_{t_1}}})$. 
The inputs of the first set also determine which set comes second, $\A_{t^{\vec{x}}_2}$, where $t^{\vec{x}}_2=t_2(\vec x_{ {\A_{t_1}}})$, whose outputs can depend on all inputs of the first and second sets; and so on, until all the sets in the partition are ordered.
As one can see, each possible vector of inputs $\vec{x}$ thus determines (in a not necessarily unique way) a given causal order for the sets of parties in $\Ptn$.

\subsection{Non-inclusion relations for $\Ptn$-causal polytopes}\label{subsec:nonincl_Pcausal}

As suggested earlier, our goal is to quantify the extent to which a noncausal resource is genuinely multipartite in terms of the number or size of the subsets one needs to consider in a partition $\Ptn$ to make a given correlation $\Ptn$-causal. 
A natural property to demand of such a quantification is that it defines nested sets of correlations: if a correlation is genuinely multipartite noncausal ``to a certain degree'', it should also be contained in the sets of ``less genuinely multipartite noncausal'' correlations (and, eventually, the set of simply noncausal correlations).
It is therefore useful, before providing the relevant definitions in the next subsections, to gather a better understanding of the inclusion relations between $\Ptn$-causal polytopes.

One might intuitively think that there should indeed be nontrivial inclusion relations among those polytopes. 
For example, one might think that a $\Ptn$-causal correlation should also be $\Ptn'$-causal if $\Ptn'$ is a ``coarse-graining'' of $\Ptn$ (i.e., $\Ptn'$ is obtained from $\Ptn$ by grouping some of its groups to define fewer but larger subsets)---or, more generally, when $\Ptn'$ contains fewer subsets than $\Ptn$, i.e.\ $|\Ptn'| < |\Ptn|$. 
This, however, is not true. 
For example, in the tripartite case, a fully causal correlation (i.e., a $\Ptn$-causal one for $\Ptn = \{\{A_1\},\{A_2\},\{A_3\}\}$) compatible with the fixed order $A_1 \prec A_2 \prec A_3$, where $A_2$ comes between $A_1$ and $A_3$, may not be $\Ptn'$-causal for $\Ptn' = \{\{A_1,A_3\},\{A_2\}\}$, since one cannot order $A_2$ with respect to $\{A_1,A_3\}$ when those are taken together. 
In fact, no nontrivial inclusion exists among $\Ptn$-causal polytopes, as established by the following theorem, proved in Appendix~\ref{Pseparationproof}.
\begin{theorem}\label{Pseparation}
Consider an $N$-partite scenario where each party has at least two possible inputs and at least two possible outputs for one value of the inputs. 
Given two distinct nontrivial\footnote{If one of the two partitions is trivial, say $\Ptn' = \{\N\}$, then the $\Ptn$-causal polytope is of course contained in the trivial $\Ptn'$-causal one (which contains all valid probability distributions). Note that for $N=2$ there is only one nontrivial partition; the theorem is thus only relevant for scenarios with $N \ge 3$.} partitions $\Ptn$ and $\Ptn'$ of $\N$ with $|\Ptn|,|\Ptn'| > 1$, the $\Ptn$-causal polytope is not contained in the $\Ptn'$-causal one, nor vice versa.
\end{theorem}

One may also ask whether, for a given $\Ptn$-causal correlation $P$, there always exists a partition $\Ptn'$ with $2 \le |\Ptn'| < |\Ptn|$ such that $P$ is also $\Ptn'$-causal (recall that the case $|\Ptn'| = 1$ is trivial). 
The answer is negative when mixtures of different causal orders are involved: e.g., in the tripartite case with $\Ptn = \{\{A_1\},\{A_2\},\{A_3\}\}$, a fully causal correlation of the form $P = \frac{1}{6} \, (P_{A_1 \prec A_2 \prec A_3} + P_{A_1 \prec A_3 \prec A_2} + P_{A_2 \prec A_1 \prec A_3} + P_{A_2 \prec A_3 \prec A_1} + P_{A_3 \prec A_1 \prec A_2} + P_{A_3 \prec A_2 \prec A_1})$, where each correlation in the sum is compatible with the corresponding causal order, may not be $\Ptn'$-causal for any $\Ptn'$ of the form $\Ptn' = \{ \{A_i, A_j\}, \{A_k\} \}_{i \neq j \neq k}$, as there is always a term in $P$ above for which $A_k$ comes between $A_i$ and $A_j$. 
For an explicit example one can take the correlation $P$ above to be a mixture of 6 correlations $P_{\Ptn,\sigma}^\text{det}$ introduced in Appendix~\ref{Pseparationproof}.%
\footnote{To see that $P$ thus defined is indeed not $\Ptn'$-causal for any such bipartition, first note that, by symmetry, it suffices to show it is not $\Ptn'$-causal for $\Ptn'=\{\{A_1\},\{A_2,A_3\}\}$.
One can readily show that all such $\Ptn'$-causal inequalities must obey the \mbox{LGYNI}-type inequality $P_{A_1}(1|100)+P_{A_2A_3}(11|011)-P_{A_1A_2A_3}(111|111)\ge 0$ (which, moreover, is a facet of the $\Ptn'$-causal polytope).
It is easily verified that $P$ violates this inequality with the left-hand side obtaining the value $-1/3$.}

The above results tell us that $\Ptn$-causal polytopes do not really define useful classes to directly quantify how genuinely multipartite the noncausality of a correlation is. 
One may wonder whether considering convex hulls of $\Ptn$-causal polytopes allows one to avoid these issues.
For example, is it the case that any $\Ptn$-causal correlation $P$ is contained in the convex hull of all $\Ptn_j'$-causal correlations for all partitions $\Ptn_j'$ with a fixed value of $|\Ptn_j'| = \mathfrak{m}' < |\Ptn|$?%
\footnote{Note that a convex combination of $\Ptn_j'$-causal correlations for various partitions $\Ptn_j'$ with a fixed number of subsets $|\Ptn_j'| = \mathfrak{m}'$ is not necessarily $\Ptn'$-causal for any single partition $\Ptn'$ with the same value of $|\Ptn'| = \mathfrak{m}'$.}
For $\mathfrak{m}'=1$ this is trivial, and this remains true for $\mathfrak{m}'=2$: any $\Ptn$-causal correlation $P$ can be decomposed as a convex combination of $\Ptn_j'$-causal correlations for various partitions $\Ptn_j'$ with $|\Ptn_j'| = 2$. Eq.~\eqref{eqdef:P_causal} is indeed such a decomposition, with the partitions $\Ptn_\ell' = \{\A_\ell, \N \backslash \A_\ell\}$. This is also true, for any value of $\mathfrak{m}'$, for $\Ptn$-causal correlations that are compatible with a fixed causal order between the subsets in $\Ptn$ (or convex mixtures thereof): indeed, such a correlation is also $\Ptn'$-causal for any coarse-grained partition $\Ptn'$ of $\Ptn$ where consecutive subsets (as per the causal order in question, or per each causal order in a convex mixture) of $\Ptn$ are grouped together.
However, this is not true in general for $\mathfrak{m}' > 2$ when dynamical causal orders are involved. 
It is indeed possible to find a 4-partite, fully causal correlation that cannot be expressed as a convex combination of $\Ptn_j'$-causal correlations with all $|\Ptn_j'| = 3$; an explicit counterexample is presented in Appendix~\ref{dynamicalcounter}.

From these observations we conclude that, although grouping parties into $\mathfrak{m}$ subsets seems to be a stronger constraint than grouping parties into some $\mathfrak{m}' < \mathfrak{m}$ subsets, the fact that a correlation is $\Ptn$-causal for some $|\Ptn| = \mathfrak{m} \ge 4$ (or more generally, that it is a convex combination of various $\Ptn_j$-causal correlations with all $|\Ptn_j| = \mathfrak{m} \ge 4$) does not guarantee that it is also $\Ptn'$-causal for some $|\Ptn'| = \mathfrak{m}' < \mathfrak{m}$---unless $\mathfrak{m}'=2$ (or $\mathfrak{m}' = 1$, trivially)---nor that it can be decomposed as a convex combination of $\Ptn_j'$-causal correlations with all $|\Ptn_j'| = \mathfrak{m}'$. In particular, fully causal correlations may not be $\Ptn'$-causal for any $\Ptn'$ with $2 < |\Ptn'| < N$, or convex combinations of such $\Ptn'$-causal correlations.
This remark motivates the definitions in the next subsection.

\subsection{$M$-causal correlations}

\subsubsection{Definition and characterisation}

With the previous discussion in mind, we propose the following definition, as a first refinement between the definitions of fully causal and 2-causal correlations.

\begin{definition}[$M$-causal correlations] \label{def:M_causal}
An $N$-partite correlation is said to be $M$-causal (for $1 \le M \le N$) if and only if it is a convex combination of $\Ptn$-causal correlations, for various partitions $\Ptn$ of $\N$ into $|\Ptn| \ge M$ subsets.

More explicitly: $P$ is $M$-causal if and only if it can be decomposed as
\begin{equation}
P(\vec a|\vec x) = \sum_{\Ptn: \, |\Ptn| \ge M} q_\Ptn \, P_\Ptn(\vec a|\vec x), \label{eqdef:M_causal}
\end{equation}
where the sum is over all partitions $\Ptn$ of $\N$ into $M$ subsets or more, with $q_\Ptn \ge 0$ for each $\Ptn$, $\sum_{\Ptn} q_\Ptn = 1$, and where each $P_\Ptn(\vec a|\vec x)$ is a $\Ptn$-causal correlation.
\end{definition}

For $M=1$, any correlation is trivially 1-causal, since for $\Ptn = \{\N\}$ any correlation is $\Ptn$-causal. 
For $M=N$, the definition of $M$-causal correlations above is equivalent to that of fully causal correlations, Definition~\ref{def:fully_causal}~\cite{oreshkov16,abbott16}.

For $M=2$, the above definition is equivalent to that of 2-causal correlations as introduced through Definition~\ref{def:2_causal}. 
To see this, recall first (from the discussion in the previous subsection), that any $\Ptn$-causal correlation with $|\Ptn| \ge 2$ can be written as a convex combination of some $\Ptn'$-causal correlations, for various bipartitions $\Ptn'$ with $|\Ptn'| = 2$. 
It follows that, for $M=2$, it would be equivalent to have the condition $|\Ptn| = 2$ instead of $|\Ptn| \ge 2$ in Definition~\ref{def:M_causal} of $M$-causal correlations. 
Definition~\ref{def:2_causal} is then recovered when writing the bipartitions in the decomposition as $\Ptn = \{ \K, \N \backslash \K \}$, using Eq.~\eqref{eqdef:P_causal} from the definition of $\Ptn$-causal correlations, and rearranging the terms in the decomposition. 
Hence, Definition~\ref{def:2_causal} is in fact equivalent to saying that 2-causal correlations are those that can be written as a convex mixture of $\Ptn$-causal correlations, for different partitions $\Ptn$ of $\N$ into $|\Ptn| \ge 2$ subsets, thus justifying further our definition of genuinely $N$-partite noncausal correlations as those that cannot be written as such a convex mixture (or equivalently, those that are not $M$-causal for any $M>1$).
Note that since we used the constraint $|\Ptn| \ge M$ rather than $|\Ptn| = M$ in Eq.~\eqref{eqdef:M_causal},%
\footnote{Replacing the condition $|\Ptn| \ge M$ by $|\Ptn| = M$ in Definition~\ref{def:M_causal} for arbitrary $M$, we could define ``$({=}M)$-causal correlations'', which would be distinct from $M$-causal correlations for $2<M<N$. 
We would also have that ``$({=}M)$-causal correlations'' form a convex polytope; however, the various ``$({=}M)$-causal polytopes'' would not necessarily be included in one another for distinct values of $M$, as discussed in the previous subsection. 
\label{ftnote:eqK}} 
our definition establishes a hierarchy of correlations as desired, with $M$-causal $\Rightarrow$ $M'$-causal if $M \ge M'$.

With the above definition of $M$-causal correlations, we have the following:
\begin{theorem}
\label{thm:M_causal_polytope}
For any given value of $M$ (with $1 \le M \le N$), the set of $M$-causal correlations forms a convex polytope, whose (finitely many) extremal points correspond to deterministic $\Ptn$-causal correlations, for all possible partitions $\Ptn$ with $|\Ptn| \ge M$---that is, deterministic $M$-causal correlations.
\end{theorem}

\begin{proof}
According to Eq.~\eqref{eqdef:M_causal}, the set of $M$-causal correlations is the convex hull of the polytopes of $\Ptn$-causal correlations with $|\Ptn| \ge M$. Since there is a finite number of such polytopes, the set of $M$-causal correlations is itself a convex polytope; its extremal points are those of the various $\Ptn$-causal polytopes with $|\Ptn| \ge M$, namely deterministic $\Ptn$-causal correlations (see Theorem~\ref{thm:P_causal_polytope}).
\end{proof}

\begin{figure}
	\begin{center}
		\includegraphics[width=.8\columnwidth]{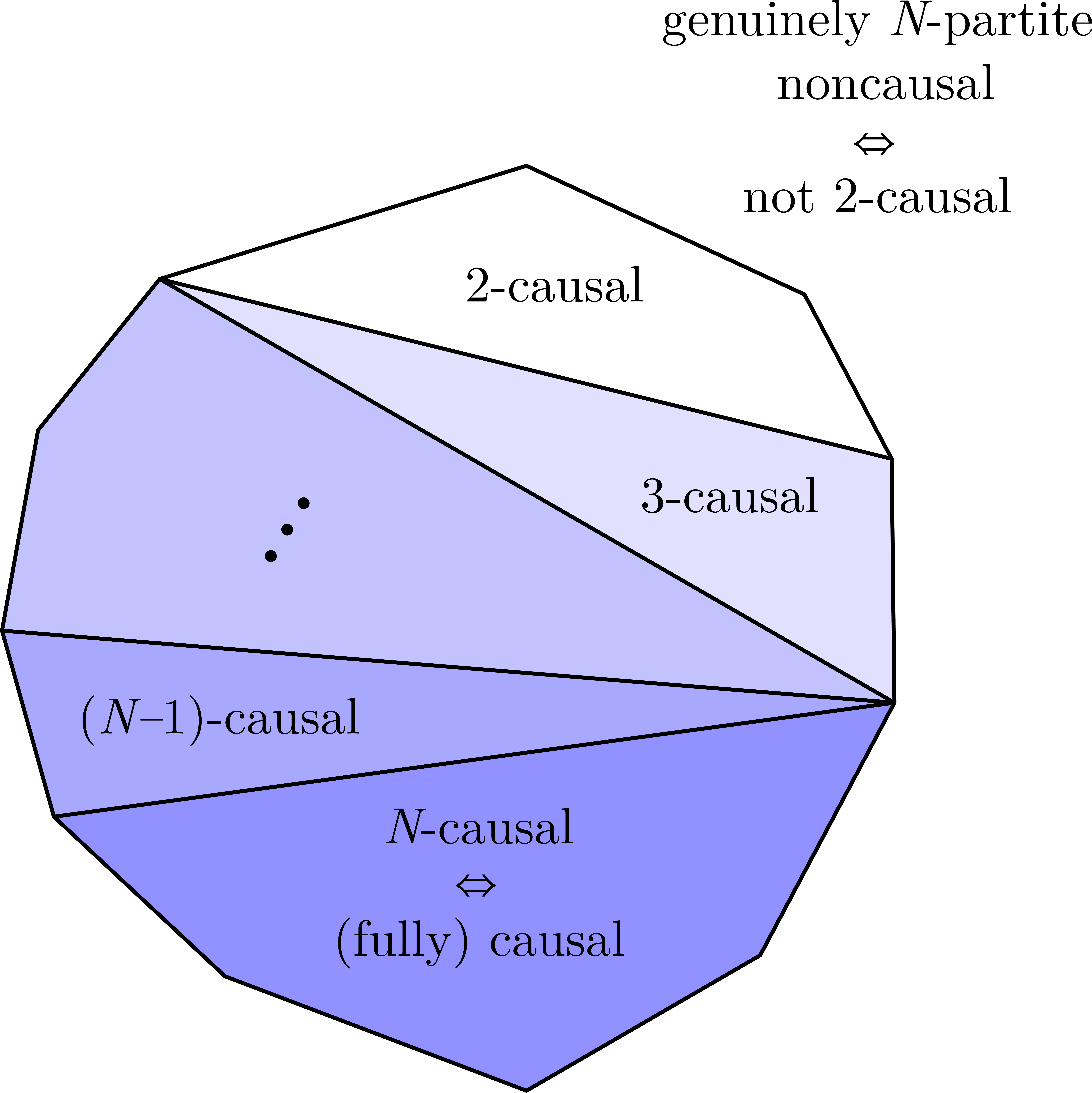}
	\end{center}
\caption{Sketch of the full hierarchy of nested $M$-causal polytopes, that allows one to refine the characterisation represented in Fig.~\ref{fig:2polytopes}. The vertices of each $M$-causal polytope correspond to deterministic $M$-causal correlations, and their facets correspond to $M$-causal inequalities.
An $M$-causal correlation is also $M'$-causal for any $M' \le M$; the largest $M$ for which a correlation $P$ is $M$-causal quantifies how genuinely multipartite its noncausality is.}
\label{fig:hierarchy}
\end{figure}

We thus obtain a family of convex polytopes---which we shall call \emph{$M$-causal polytopes}---included in one another, see Fig.~\ref{fig:hierarchy}. 
The facets of these polytopes are \emph{$M$-causal inequalities}, which define a hierarchy of criteria: e.g., if all $M$-causal inequalities are satisfied, then so are all $M'$-causal inequalities if $M' \le M$---or equivalently: if some $M'$-causal inequality is violated, then some $M$-causal inequality must also be violated if $M \ge M'$.
Given a correlation $P$, one can in principle test to which set it belongs. The largest $M$ for which $P$ is $M$-causal can be used as a measure of how genuinely multipartite its noncausality is: it means that $P$ can be obtained as a convex combination of $\Ptn$-causal correlations with all $|\Ptn| \ge M$, but not with all $|\Ptn| > M$---indeed, if $M<N$ then $P$ violates some $M'$-causal inequality for any $M'>M$ (with $M' \le N$).
If that $M$ is $1$, $P$ is a genuinely $N$-partite noncausal correlation; if it is $N$, then $P$ is fully causal, hence it displays no noncausality (genuinely multipartite or not).

\subsubsection{A family of $M$-causal inequalities}

The general $N$-partite 2-causal inequality~\eqref{eq:genIneq2} can easily be modified to give an $M$-causal inequality that is valid---although not tight in general, as observed before---for all $N$ and $M$ (with $1 \le M \le N$), simply by changing the bound.
Indeed, this bound is derived from the largest possible number of pairs of parties that can be in a single subset of a given partition, and this can easily be recalculated for $M$-subset partitions rather than bipartitions. We thus obtain that
\begin{equation}\label{eq:McausalIneq}
	J_2(N)=\sum_{\{i,j\}\subset \N} L_N(i,j)\ \ge \ -\binom{N-M+1}{2}
\end{equation}
for any $M$-causal correlation.
This updated bound is proved in Appendix~\ref{appndx:genIneq}.

As for Eq.~\eqref{eq:genIneq2} it is easy to see that, for any $N,M$, there are $M$-causal correlations saturating the inequality~\eqref{eq:McausalIneq}.
Consider, for instance, the partition $\Ptn=\{\A_1,\dots,\A_M\}$ of $\N$ with $\A_1=\{1,\dots,N-M+1\}$ and $\A_\ell=\{N-M+\ell\}$ for $2\le \ell \le M$, and take the distribution
\begin{equation}\label{eq:McausalCorrelation}
	P(\vec{a}|\vec{x})= \delta_{\vec{a}_{\A_1},f(\vec{x}_{\A_1})} \, \delta_{\vec{a}_{\N\backslash\A_1},\vec{0}},
\end{equation}
where we use the same function $f$ as in Eq.~\eqref{eq:genIneqViolatingDist}. 
Analogous reasoning shows that this correlation indeed reaches the bound~\eqref{eq:McausalIneq}.

Since this (reachable) lower bound is different for each possible value of $M$, this implies, in particular, that (for the $N$-partite lazy scenario) all the inclusions $N$-causal $\subset (N{-}1)$-causal $\subset \cdots \subset$ 3-causal $\subset$ 2-causal in the hierarchy of $M$-causal polytopes are strict.
In fact, redas for inequalities~\eqref{eq:genIneq1} and~\eqref{eq:genIneq2} (see Footnote~\ref{footnote:genScenarios}), the proof of Eq.~\eqref{eq:McausalIneq} holds in any nontrivial scenario (with arbitrarily many inputs and outputs), of which the lazy scenario is the simplest example for all $N$.
Moreover, one can saturate it in such scenarios by trivially extending the $M$-causal correlation~\eqref{eq:McausalCorrelation} (e.g., by producing a constant output on all other inputs) and thus these inclusions are strict in general.

\subsection{Size-$S$-causal correlations}

In the previous subsection we used the number of subsets needed in a partition to quantify how genuinely multipartite the noncausality of a correlation is. 
Here we present an alternative quantification, based on the size of the biggest subset in a partition, rather than the number of subsets.

Intuitively, the bigger the subsets in a partition $\Ptn$ needed to reproduce a correlation, the more genuinely multipartite noncausal the corresponding $\Ptn$-causal correlations are. 
However, the discussion of Sec.~\ref{subsec:nonincl_Pcausal} implies that, as was the case with $M$-causal correlations, it is not sufficient to simply ask whether a given correlation is $\Ptn$-causal for some partition $\Ptn$ with subsets of a particular size.
We therefore focus on classes of correlations that can be written as mixtures of $\Ptn$-causal ones whose largest subset is not larger than some number $S$. 
For convenience, we introduce the notation
\begin{equation}
\mathfrak{s}(\Ptn)\coloneqq\max_{\A\in\Ptn}\left|\A\right|.
\label{eq:maxSizeDef}
\end{equation}
We then take the following definition:

\begin{definition}[Size-$S$-causal correlations] \label{def:S_causal}
An $N$-partite correlation is said to be size-$S$-causal (for $1 \le S \le N$) if and only if it is a convex combination of $\Ptn$-causal correlations, for various partitions $\Ptn$ whose subsets are no larger than $S$.

More explicitly: $P$ is size-$S$-causal if and only if it can be decomposed as
\begin{equation}
P(\vec a|\vec x) = \hspace{-2mm}\sum_{\Ptn: \, \mathfrak{s}(\Ptn) \leq S}\hspace{-2mm} q_\Ptn \, P_\Ptn(\vec a|\vec x), \label{eqdef:S_causal}
\end{equation}
where the sum is over all partitions $\Ptn$ of $\N$ with no subset of size larger than $S$, with $q_\Ptn \ge 0$ for each $\Ptn$, $\sum_{\Ptn} q_\Ptn = 1$, and where each $P_\Ptn(\vec a|\vec x)$ is a $\Ptn$-causal correlation.
\end{definition}

Any $N$-partite correlation is trivially size-$N$-causal, while size-$1$-causal correlations coincide with fully causal correlations. Furthermore, noting that $\mathfrak{s}(\Ptn) \leq N-1$ if and only if $|\Ptn| \ge 2$, we see that the set of size-$(N{-}1)$-causal correlations coincides with that of $2$-causal correlations. 
Hence, the definition of size-$S$-causal correlations is another possible generalisation of that of 2-causal ones.
From this new perspective, $2$-causal correlations can be seen as those that can be realised using (probabilistic mixtures of) noncausal resources available to groups of parties of size $N{-}1$ or less. 
This further strengthens the definition of $2$-causal correlations as the largest set of correlations that do not possess genuinely $N$-partite noncausality.

Without repeating in full detail, it is clear that size-$S$-causal correlations define a structure similar to that of $M$-causal correlations: for each $S$, size-$S$-causal correlations define \emph{size-$S$-causal polytopes} whose vertices are deterministic size-$S$-causal correlations and whose facets define \emph{size-$S$-causal inequalities}. 
For $S \le S'$, all size-$S$-causal correlations are also size-$S'$-causal, so that the various size-$S$-causal polytopes are included in one another. 
The lowest $S$ for which a correlation is size-$S$-causal also provides a measure of how genuinely multipartite the corresponding noncausal resource is, distinct to that defined by $M$-causal correlations.
 
It is also possible here to generalise inequality~\eqref{eq:genIneq2} to size-$S$-causal correlations by changing the bound. As proven in Appendix~\ref{appndx:genIneq}, we thus obtain the size-$S$-causal inequality
\begin{equation}
J_2(N)\geq - \left\lfloor \frac{N}{S}\right\rfloor \binom{S}{2} - \binom{N - \left\lfloor \frac{N}{S}\right\rfloor S}{2}
\label{sizeSbound}
\end{equation}
(where $\left\lfloor x \right\rfloor$ denotes the largest integer smaller than or equal to $x$).
Although, once again, this inequality is not tight in the sense that it does not define a facet of the size-$S$-causal polytope, its lower bound can be saturated by a size-$S$-causal correlation for each value of $S$, for instance by considering the partition $\Ptn = \{\A_1,\ldots,\A_{\left\lfloor \frac{N}{S}\right\rfloor} (,\A_{\left\lfloor \frac{N}{S}\right\rfloor+1})\}$ of $\N$ into $\left\lfloor \frac{N}{S}\right\rfloor$ groups of $S$ parties, and (if $N$ is not a multiple of $S$) a last group with the remaining $N - \left\lfloor \frac{N}{S}\right\rfloor S$ parties, and by taking the deterministic correlation
\begin{equation}
	P(\vec{a}|\vec{x})= \prod_\ell \delta_{\vec{a}_{\A_\ell},f(\vec{x}_{\A_\ell})}
\end{equation}
(with again the same function $f$ as in Eq.~\eqref{eq:genIneqViolatingDist}). 
Since the (reachable) lower bounds in Eq.~\eqref{sizeSbound} are different for all possible values of $S$, this implies, again, that all the inclusions size-$1$-causal $\subset$ size-$2$-causal $\subset \cdots \subset$ size-$(N{-}1)$-causal in the hierarchy of size-$S$-causal polytopes are strict in general.

\subsection{Comparing the polytopes of $M$-causal and size-$S$-causal correlations} \label{sec:compare_M_sizeS}

A relation between $M$-causal and size-$S$-causal correlations can be established through the following, straightforwardly verifiable, inequalities:
\begin{equation} \label{eq:ineqs_S_M}
|\Ptn| - 1 +\mathfrak{s}(\Ptn) \leq N \leq |\Ptn| \, \mathfrak{s}(\Ptn).
\end{equation}
From these it follows that, for $N$ parties: 
\pagebreak
\begin{theorem}\label{thm:inclusions} $ $
\begin{itemize}
	\item If a correlation is $M$-causal, then it is size-$S$-causal for all $S \ge N{-}M{+}1$.
	\item If a correlation is size-$S$-causal, then it is $M$-causal for all $M \le \left\lceil \frac{N}{S}\right\rceil$ (where $\left\lceil x \right\rceil$ denotes the smallest integer larger than or equal to $x$). 
\end{itemize}
\end{theorem}

It is furthermore possible to show that the inclusion relations between $M$-causal and size-$S$-causal polytopes implied by Theorem~\ref{thm:inclusions} are complete, in the sense that no other inclusion exists that is not implied by the theorem.
We prove this in Appendix~\ref{app_sec:inclusions}.
Together with the respective inclusion relations of each hierarchy separately, this result thus fully characterises the inclusion relations of all the classes of noncausal correlations that we introduced; the situation is illustrated in Fig.~\ref{fig:inclusions} for the 6-partite case as an example.

\begin{figure}%
	\begin{center}
		\includegraphics[width=0.45\columnwidth]{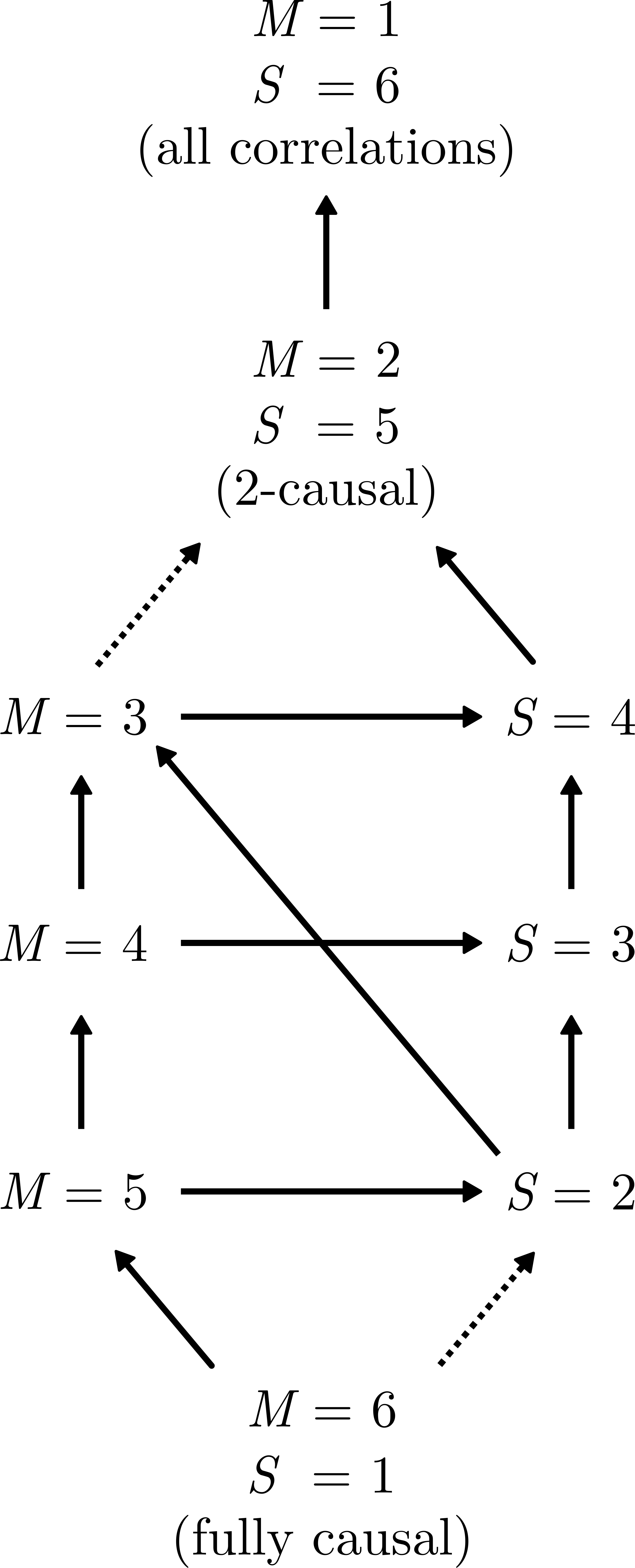}
	\end{center}
\caption{Diagram of polytope inclusion relations for $N=6$ parties. Each node in the diagram represents the polytope of $M$-causal or size-$S$-causal correlations for the corresponding value of $M$ or $S$. Arrows represent (strict) inclusion relations: the polytope at the start of an arrow is included in the polytope at the end of the arrow. The paths ascending the left and right sides of the diagram show the inclusion relations implied by the fact that the $M$-causal and size-$S$-causal polytopes respectively form strict hierarchies. The ``cross-arrows'' between these paths are the inclusions implied by Theorem~\ref{thm:inclusions}. The dotted arrows can be implied by transitivity from other arrows and are thus redundant. The inclusion relations shown in the diagram are complete, in the sense that there is no inclusion not represented by an arrow or a sequence of arrows.
}%
\label{fig:inclusions}%
\end{figure}

\section{Discussion}

The possibility that nature might allow for correlations incompatible with a definite causal order opens exciting questions. 
It has been suggested that such correlations might arise in the context of quantum theories of gravity~\cite{hardy2007towards} or in space-time geometries with closed time-like curves~\cite{baumeler2017reversible,araujo2017quantum}, although these possibilities, like that of observing noncausal correlations in laboratory experiments, are as yet unverified. 
Motivated by the fact that noncausal resources exhibit interesting new features in multipartite scenarios~\cite{baumeler13,baumeler14,oreshkov16,abbott16}, we aimed here to clarify when noncausal correlations can be considered to be a genuinely multipartite resource.

In addressing this task, we first proposed a criterion to decide whether a given correlation shared by $N$ parties is ``genuinely $N$-partite noncausal''---i.e., its noncausality is indeed a genuinely $N$-partite resource---or not. We then refined our approach into two distinct criteria quantifying the extent to which the noncausality of a correlation is a genuinely multipartite resource.
Both criteria are based around asking whether the correlation under consideration is compatible with certain subsets being grouped together---which are thus able to share arbitrary noncausal resources---and with a well-defined causal order existing between these groups of parties.
The first criterion is based on the largest number $M$ of such subsets that can be causally ordered while reproducing the correlation in question: the smaller $M$, the more genuinely multipartite the noncausality exhibited by the correlation.
If $M=1$, then no subset of parties has a well-defined causal relation with any other, and the correlation is genuinely $N$-partite noncausal.
The second criterion instead looks at how large the subsets that can be causally ordered are: if an $N$-partite correlation can be reproduced with subsets containing no more than $S \le N$ parties, then $S$-partite noncausal resources are sufficient to reproduce the correlation.
Thus, the larger $S$ required, the more genuinely multipartite the correlation.
If $S=N$, then again the correlation is genuinely $N$-partite noncausal.
Although these two criteria define different classes of correlations in general, they coincide on the edge cases and thus lead to exactly the same definition of genuinely $N$-partite noncausal correlations, adding support to the robustness of our definition.
It nonetheless remains to be seen as to which measure of genuine multipartiteness is the most appropriate (or, in what situations one is more pertinent than the other).
 
All the classes of correlations we introduced through these criteria conveniently form polytopes, whose vertices are deterministic correlations and whose facets define different classes of inequalities. 
Of particular interest are the ``2-causal'' correlations, which are the most general correlations that are not genuinely $N$-partite noncausal. 
We completely characterised the 2-causal polytope for the simplest nontrivial tripartite scenario and found that almost all of the 473 nontrivial classes of 2-causal inequalities can be violated by process matrix correlations. 
However, we were unable to find any violation for 2 of those inequalities; this stands in contrast to previous studies of causal inequalities, where violations with process matrices were always found\footnote{At least for standard causal inequalities that bound probabilities directly; for \emph{entropic} causal inequalities, which only provide a relaxed characterisation of the set of causal correlations, no violations were found so far~\cite{miklin17}. It would nevertheless also be interesting to investigate how genuinely multipartite noncausality can be characterised with the entropic approach.}~\cite{oreshkov12,baumeler13,baumeler14,branciard16,abbott16}.
Although it remains to be confirmed whether this is simply a failure of the search method we used, we provided some intuition why such a violation would in fact be a surprise.

Our definition of genuinely $N$-partite noncausality is analogous to the corresponding notion for nonlocality originally due to Svetlichny~\cite{svetlichny87,seevinck02,collins02}. 
It is known, however, that that notion harbours some issues: for example, it is not robust under local operations, a necessary requirement for an operational resource theory of nonlocality~\cite{gallego12,bancal13}. 
In order to overcome these issues, additional constraints must be imposed on the correlations shared by subsets of parties when defining correlations that are not genuinely multipartite nonlocal.
In the case of noncausality, however, there appears to be no clear reason to impose any additional such constraints. 
For nonlocal resources, issues arise in particular from the possibility that different parties might access the resource at different times, with an earlier party then communicating with a later one. 
This type of issue is not pertinent for noncausal resources, where the causal order (be it definite or indefinite) between parties is determined by the resource itself, and additional communication beyond what the resource specifies seems to fall outside the relevant framework.
More generally, however, an operational framework and understanding of the relevant ``free operations'' for noncausal resources remains to be properly developed.

Finally, in this paper we only considered correlations from a fully theory- and device-independent perspective; it would be interesting to develop similar notions within specific physical theories like the process matrix framework, where quantum theory holds locally for each party. 
Process matrices that cannot be realised with a definite causal order are called causally nonseparable~\cite{oreshkov12}, and it would be interesting to study a notion of genuinely multipartite causal nonseparability. 
It should, however, be noted that different possible notions of multipartite causal (non)separability have been proposed~\cite{araujo15,oreshkov16}, so a better understanding of their significance would be necessary in order to extend the notions we have developed here to that framework.

\section*{Acknowledgements}

A.A.A., J.W.\ and C.B.\ acknowledge support through the `Retour Post-Doctorants' program (ANR-13-PDOC-0026) of the French National Research Agency. 
F.C.\ acknowledges support through an Australian Research Council Discovery Early Career Researcher Award (DE170100712). 
This publication was made possible through the support of a grant from the John Templeton Foundation. 
The opinions expressed in this publication are those of the authors and do not necessarily reflect the views of the John Templeton Foundation. 
F.C.\ acknowledges the traditional owners of the land on which the University of Queensland is situated, the Turrbal and Jagera people.

\nocite{apsrev41Control} 

\bibliographystyle{apsrev4-1_modified}
\bibliography{bib_genuine_multi_noncausal}

\begin{thebibliography}{24}%
\makeatletter
\providecommand \@ifxundefined [1]{%
 \@ifx{#1\undefined}
}%
\providecommand \@ifnum [1]{%
 \ifnum #1\expandafter \@firstoftwo
 \else \expandafter \@secondoftwo
 \fi
}%
\providecommand \@ifx [1]{%
 \ifx #1\expandafter \@firstoftwo
 \else \expandafter \@secondoftwo
 \fi
}%
\providecommand \natexlab [1]{#1}%
\providecommand \enquote  [1]{#1}%
\providecommand \bibnamefont  [1]{#1}%
\providecommand \bibfnamefont [1]{#1}%
\providecommand \citenamefont [1]{#1}%
\providecommand \href@noop [0]{\@secondoftwo}%
\providecommand \href [0]{\begingroup \@sanitize@url \@href}%
\providecommand \@href[1]{\@@startlink{#1}\@@href}%
\providecommand \@@href[1]{\endgroup#1\@@endlink}%
\providecommand \@sanitize@url [0]{\catcode `\\12\catcode `\$12\catcode
  `\&12\catcode `\#12\catcode `\^12\catcode `\_12\catcode `\%12\relax}%
\providecommand \@@startlink[1]{}%
\providecommand \@@endlink[0]{}%
\providecommand \url  [0]{\begingroup\@sanitize@url \@url }%
\providecommand \@url [1]{\endgroup\@href {#1}{\urlprefix }}%
\providecommand \urlprefix  [0]{URL }%
\providecommand \Eprint [0]{\href }%
\providecommand \doibase [0]{http://dx.doi.org/}%
\providecommand \selectlanguage [0]{\@gobble}%
\providecommand \bibinfo  [0]{\@secondoftwo}%
\providecommand \bibfield  [0]{\@secondoftwo}%
\providecommand \translation [1]{[#1]}%
\providecommand \BibitemOpen [0]{}%
\providecommand \bibitemStop [0]{}%
\providecommand \bibitemNoStop [0]{.\EOS\space}%
\providecommand \EOS [0]{\spacefactor3000\relax}%
\providecommand \BibitemShut  [1]{\csname bibitem#1\endcsname}%
\let\auto@bib@innerbib\@empty
\bibitem [{\citenamefont {Reichenbach}(1956)}]{Reichenbachbook}%
  \BibitemOpen
  \bibfield  {author} {\bibinfo {author} {\bibfnamefont {H.}~\bibnamefont
  {Reichenbach}},\ }\href@noop {} {\emph {\bibinfo {title} {The direction of
  time}}}\ (\bibinfo  {publisher} {University of California Press},\ \bibinfo
  {address} {Berkeley},\ \bibinfo {year} {1956})\BibitemShut {NoStop}%
\bibitem [{\citenamefont {Pearl}(2009)}]{Pearlbook}%
  \BibitemOpen
  \bibfield  {author} {\bibinfo {author} {\bibfnamefont {J.}~\bibnamefont
  {Pearl}},\ }\href@noop {} {\emph {\bibinfo {title} {Causality}}}\ (\bibinfo
  {publisher} {Cambridge University Press, Cambridge},\ \bibinfo {year}
  {2009})\BibitemShut {NoStop}%
\bibitem [{\citenamefont {{Brukner}}(2014)}]{brukner14}%
  \BibitemOpen
  \bibfield  {author} {\bibinfo {author} {\bibfnamefont {{\v C}.}~\bibnamefont
  {{Brukner}}},\ }\bibfield  {title} {\enquote {\bibinfo {title} {Quantum
  causality},}\ }\href {\doibase 10.1038/nphys2930} {\bibfield  {journal}
  {\bibinfo  {journal} {Nat. Phys.}\ }\textbf {\bibinfo {volume} {10}},\
  \bibinfo {pages} {259--263} (\bibinfo {year} {2014})}\BibitemShut {NoStop}%
\bibitem [{\citenamefont {Oreshkov}\ \emph {et~al.}(2012)\citenamefont
  {Oreshkov}, \citenamefont {Costa},\ and\ \citenamefont
  {Brukner}}]{oreshkov12}%
  \BibitemOpen
  \bibfield  {author} {\bibinfo {author} {\bibfnamefont {O.}~\bibnamefont
  {Oreshkov}}, \bibinfo {author} {\bibfnamefont {F.}~\bibnamefont {Costa}}, \
  and\ \bibinfo {author} {\bibfnamefont {{\v{C}}.}~\bibnamefont {Brukner}},\
  }\bibfield  {title} {\enquote {\bibinfo {title} {Quantum correlations with no
  causal order},}\ }\href {\doibase 10.1038/ncomms2076} {\bibfield  {journal}
  {\bibinfo  {journal} {Nat. Commun.}\ }\textbf {\bibinfo {volume} {3}},\
  \bibinfo {pages} {1092} (\bibinfo {year} {2012})}\BibitemShut {NoStop}%
\bibitem [{\citenamefont {Branciard}\ \emph {et~al.}(2016)\citenamefont
  {Branciard}, \citenamefont {Ara\'{u}jo}, \citenamefont {Feix}, \citenamefont
  {Costa},\ and\ \citenamefont {Brukner}}]{branciard16}%
  \BibitemOpen
  \bibfield  {author} {\bibinfo {author} {\bibfnamefont {C.}~\bibnamefont
  {Branciard}}, \bibinfo {author} {\bibfnamefont {M.}~\bibnamefont
  {Ara\'{u}jo}}, \bibinfo {author} {\bibfnamefont {A.}~\bibnamefont {Feix}},
  \bibinfo {author} {\bibfnamefont {F.}~\bibnamefont {Costa}}, \ and\ \bibinfo
  {author} {\bibfnamefont {{\v{C}}.}~\bibnamefont {Brukner}},\ }\bibfield
  {title} {\enquote {\bibinfo {title} {The simplest causal inequalities and
  their violation},}\ }\href {\doibase 10.1088/1367-2630/18/1/013008}
  {\bibfield  {journal} {\bibinfo  {journal} {New J. Phys.}\ }\textbf {\bibinfo
  {volume} {18}},\ \bibinfo {pages} {013008} (\bibinfo {year}
  {2016})}\BibitemShut {NoStop}%
\bibitem [{\citenamefont {Baumeler}\ and\ \citenamefont
  {Wolf}(2014)}]{baumeler13}%
  \BibitemOpen
  \bibfield  {author} {\bibinfo {author} {\bibfnamefont {{\"{A}}.}~\bibnamefont
  {Baumeler}}\ and\ \bibinfo {author} {\bibfnamefont {S.}~\bibnamefont
  {Wolf}},\ }\bibfield  {title} {\enquote {\bibinfo {title} {Perfect signaling
  among three parties violating predefined causal order},}\ }in\ \href
  {\doibase 10.1109/ISIT.2014.6874888} {\emph {\bibinfo {booktitle} {2014 IEEE
  International Symposium on Information Theory (ISIT)}}}\ (\bibinfo
  {publisher} {IEEE},\ \bibinfo {address} {Piscataway, NJ},\ \bibinfo {year}
  {2014})\ pp.\ \bibinfo {pages} {526--530}\BibitemShut {NoStop}%
\bibitem [{\citenamefont {Baumeler}\ \emph {et~al.}(2014)\citenamefont
  {Baumeler}, \citenamefont {Feix},\ and\ \citenamefont {Wolf}}]{baumeler14}%
  \BibitemOpen
  \bibfield  {author} {\bibinfo {author} {\bibfnamefont {{\"{A}}.}~\bibnamefont
  {Baumeler}}, \bibinfo {author} {\bibfnamefont {A.}~\bibnamefont {Feix}}, \
  and\ \bibinfo {author} {\bibfnamefont {S.}~\bibnamefont {Wolf}},\ }\bibfield
  {title} {\enquote {\bibinfo {title} {Maximal incompatibility of locally
  classical behavior and global causal order in multi-party scenarios},}\
  }\href {\doibase 10.1103/PhysRevA.90.042106} {\bibfield  {journal} {\bibinfo
  {journal} {Phys. Rev. A}\ }\textbf {\bibinfo {volume} {90}},\ \bibinfo
  {pages} {042106} (\bibinfo {year} {2014})}\BibitemShut {NoStop}%
\bibitem [{\citenamefont {Oreshkov}\ and\ \citenamefont
  {Giarmatzi}(2016)}]{oreshkov16}%
  \BibitemOpen
  \bibfield  {author} {\bibinfo {author} {\bibfnamefont {O.}~\bibnamefont
  {Oreshkov}}\ and\ \bibinfo {author} {\bibfnamefont {C.}~\bibnamefont
  {Giarmatzi}},\ }\bibfield  {title} {\enquote {\bibinfo {title} {Causal and
  causally separable processes},}\ }\href {\doibase
  10.1088/1367-2630/18/9/093020} {\bibfield  {journal} {\bibinfo  {journal}
  {New J. Phys.}\ }\textbf {\bibinfo {volume} {18}},\ \bibinfo {pages} {093020}
  (\bibinfo {year} {2016})}\BibitemShut {NoStop}%
\bibitem [{\citenamefont {Abbott}\ \emph {et~al.}(2016)\citenamefont {Abbott},
  \citenamefont {Giarmatzi}, \citenamefont {Costa},\ and\ \citenamefont
  {Branciard}}]{abbott16}%
  \BibitemOpen
  \bibfield  {author} {\bibinfo {author} {\bibfnamefont {A.~A.}\ \bibnamefont
  {Abbott}}, \bibinfo {author} {\bibfnamefont {C.}~\bibnamefont {Giarmatzi}},
  \bibinfo {author} {\bibfnamefont {F.}~\bibnamefont {Costa}}, \ and\ \bibinfo
  {author} {\bibfnamefont {C.}~\bibnamefont {Branciard}},\ }\bibfield  {title}
  {\enquote {\bibinfo {title} {Multipartite causal correlations: Polytopes and
  inequalities},}\ }\href {\doibase 10.1103/PhysRevA.94.032131} {\bibfield
  {journal} {\bibinfo  {journal} {Phys. Rev. A}\ }\textbf {\bibinfo {volume}
  {94}},\ \bibinfo {pages} {032131} (\bibinfo {year} {2016})}\BibitemShut
  {NoStop}%
\bibitem [{\citenamefont {Hardy}(2005)}]{hardy2005probability}%
  \BibitemOpen
  \bibfield  {author} {\bibinfo {author} {\bibfnamefont {L.}~\bibnamefont
  {Hardy}},\ }\bibfield  {title} {\enquote {\bibinfo {title} {Probability
  theories with dynamic causal structure: a new framework for quantum
  gravity},}\ }\href@noop {} {\  (\bibinfo {year} {2005})},\ \Eprint
  {http://arxiv.org/abs/gr-qc/0509120}{arXiv:gr-qc/0509120}\BibitemShut
  {NoStop}%
\bibitem [{\citenamefont {Fukuda}(2012)}]{cdd}%
  \BibitemOpen
  \bibfield  {author} {\bibinfo {author} {\bibfnamefont {K.}~\bibnamefont
  {Fukuda}},\ }\href@noop {} {\enquote {\bibinfo {title}
  {\href{https://www.inf.ethz.ch/personal/fukudak/cdd_home/}{\textsc{cdd},
  v0.94g}},}\ } (\bibinfo {year} {2012}),\ \bibinfo {note}
  {\url{https://www.inf.ethz.ch/personal/fukudak/cdd_home/}}\BibitemShut
  {NoStop}%
\bibitem [{\citenamefont {Ara\'{u}jo}\ \emph {et~al.}(2017)\citenamefont
  {Ara\'{u}jo}, \citenamefont {Feix}, \citenamefont {Navascu\'{e}s},\ and\
  \citenamefont {\v{C}. Brukner}}]{araujo17}%
  \BibitemOpen
  \bibfield  {author} {\bibinfo {author} {\bibfnamefont {M.}~\bibnamefont
  {Ara\'{u}jo}}, \bibinfo {author} {\bibfnamefont {A.}~\bibnamefont {Feix}},
  \bibinfo {author} {\bibfnamefont {M.}~\bibnamefont {Navascu\'{e}s}}, \ and\
  \bibinfo {author} {\bibnamefont {\v{C}. Brukner}},\ }\bibfield  {title}
  {\enquote {\bibinfo {title} {A purification postulate for quantum mechanics
  with indefinite causal order},}\ }\href {\doibase 10.22331/q-2017-04-26-10}
  {\bibfield  {journal} {\bibinfo  {journal} {Quantum}\ }\textbf {\bibinfo
  {volume} {1}},\ \bibinfo {pages} {10} (\bibinfo {year} {2017})}\BibitemShut
  {NoStop}%
\bibitem [{\citenamefont {Feix}\ \emph {et~al.}(2016)\citenamefont {Feix},
  \citenamefont {Ara\'{u}jo},\ and\ \citenamefont {Brukner}}]{feix16}%
  \BibitemOpen
  \bibfield  {author} {\bibinfo {author} {\bibfnamefont {A.}~\bibnamefont
  {Feix}}, \bibinfo {author} {\bibfnamefont {M.}~\bibnamefont {Ara\'{u}jo}}, \
  and\ \bibinfo {author} {\bibfnamefont {{\v{C}}.}~\bibnamefont {Brukner}},\
  }\bibfield  {title} {\enquote {\bibinfo {title} {Causally nonseparable
  processes admitting a causal model},}\ }\href {\doibase
  10.1088/1367-2630/18/8/083040} {\bibfield  {journal} {\bibinfo  {journal}
  {New J. Phys.}\ }\textbf {\bibinfo {volume} {18}},\ \bibinfo {pages} {083040}
  (\bibinfo {year} {2016})}\BibitemShut {NoStop}%
\bibitem [{\citenamefont {Ara\'{u}jo}\ \emph {et~al.}(2015)\citenamefont
  {Ara\'{u}jo}, \citenamefont {Branciard}, \citenamefont {Costa}, \citenamefont
  {Feix}, \citenamefont {Giarmatzi},\ and\ \citenamefont {Brukner}}]{araujo15}%
  \BibitemOpen
  \bibfield  {author} {\bibinfo {author} {\bibfnamefont {M.}~\bibnamefont
  {Ara\'{u}jo}}, \bibinfo {author} {\bibfnamefont {C.}~\bibnamefont
  {Branciard}}, \bibinfo {author} {\bibfnamefont {F.}~\bibnamefont {Costa}},
  \bibinfo {author} {\bibfnamefont {A.}~\bibnamefont {Feix}}, \bibinfo {author}
  {\bibfnamefont {C.}~\bibnamefont {Giarmatzi}}, \ and\ \bibinfo {author}
  {\bibfnamefont {{\v{C}}.}~\bibnamefont {Brukner}},\ }\bibfield  {title}
  {\enquote {\bibinfo {title} {Witnessing causal nonseparability},}\ }\href
  {\doibase 10.1088/1367-2630/17/10/102001} {\bibfield  {journal} {\bibinfo
  {journal} {New J.~Phys.}\ }\textbf {\bibinfo {volume} {17}},\ \bibinfo
  {pages} {102001} (\bibinfo {year} {2015})}\BibitemShut {NoStop}%
\bibitem [{\citenamefont {{Hardy}}(2007)}]{hardy2007towards}%
  \BibitemOpen
  \bibfield  {author} {\bibinfo {author} {\bibfnamefont {L.}~\bibnamefont
  {{Hardy}}},\ }\bibfield  {title} {\enquote {\bibinfo {title} {{Towards
  quantum gravity: a framework for probabilistic theories with non-fixed causal
  structure}},}\ }\href {\doibase 10.1088/1751-8113/40/12/S12} {\bibfield
  {journal} {\bibinfo  {journal} {J.~Phys. A: Math. Gen.}\ }\textbf {\bibinfo
  {volume} {40}},\ \bibinfo {pages} {3081} (\bibinfo {year}
  {2007})}\BibitemShut {NoStop}%
\bibitem [{\citenamefont {Baumeler}\ \emph {et~al.}(2017)\citenamefont
  {Baumeler}, \citenamefont {Costa}, \citenamefont {Ralph}, \citenamefont
  {Wolf},\ and\ \citenamefont {Zych}}]{baumeler2017reversible}%
  \BibitemOpen
  \bibfield  {author} {\bibinfo {author} {\bibfnamefont {{\"A}.}~\bibnamefont
  {Baumeler}}, \bibinfo {author} {\bibfnamefont {F.}~\bibnamefont {Costa}},
  \bibinfo {author} {\bibfnamefont {T.~C.}\ \bibnamefont {Ralph}}, \bibinfo
  {author} {\bibfnamefont {S.}~\bibnamefont {Wolf}}, \ and\ \bibinfo {author}
  {\bibfnamefont {M.}~\bibnamefont {Zych}},\ }\bibfield  {title} {\enquote
  {\bibinfo {title} {Reversible time travel with freedom of choice},}\
  }\href@noop {} {\  (\bibinfo {year} {2017})},\ \Eprint
  {http://arxiv.org/abs/1703.00779}{arXiv:1703.00779 [gr-qc]}\BibitemShut
  {NoStop}%
\bibitem [{\citenamefont {Ara{\'u}jo}\ \emph {et~al.}(2017)\citenamefont
  {Ara{\'u}jo}, \citenamefont {Gu{\'e}rin},\ and\ \citenamefont
  {Baumeler}}]{araujo2017quantum}%
  \BibitemOpen
  \bibfield  {author} {\bibinfo {author} {\bibfnamefont {M.}~\bibnamefont
  {Ara{\'u}jo}}, \bibinfo {author} {\bibfnamefont {P.~A.}\ \bibnamefont
  {Gu{\'e}rin}}, \ and\ \bibinfo {author} {\bibfnamefont {{\"A}.}~\bibnamefont
  {Baumeler}},\ }\bibfield  {title} {\enquote {\bibinfo {title} {Quantum
  computation with indefinite causal structures},}\ }\href {\doibase
  10.1103/PhysRevA.96.052315} {\bibfield  {journal} {\bibinfo  {journal} {Phys.
  Rev. A}\ }\textbf {\bibinfo {volume} {96}},\ \bibinfo {pages} {052315}
  (\bibinfo {year} {2017})}\BibitemShut {NoStop}%
\bibitem [{\citenamefont {Miklin}\ \emph {et~al.}(2017)\citenamefont {Miklin},
  \citenamefont {Abbott}, \citenamefont {Branciard}, \citenamefont {Chaves},\
  and\ \citenamefont {Budroni}}]{miklin17}%
  \BibitemOpen
  \bibfield  {author} {\bibinfo {author} {\bibfnamefont {N.}~\bibnamefont
  {Miklin}}, \bibinfo {author} {\bibfnamefont {A.~A.}\ \bibnamefont {Abbott}},
  \bibinfo {author} {\bibfnamefont {C.}~\bibnamefont {Branciard}}, \bibinfo
  {author} {\bibfnamefont {R.}~\bibnamefont {Chaves}}, \ and\ \bibinfo {author}
  {\bibfnamefont {C.}~\bibnamefont {Budroni}},\ }\bibfield  {title} {\enquote
  {\bibinfo {title} {The entropic approach to causal correlations},}\ }\href
  {\doibase 10.1088/1367-2630/aa8f9f} {\bibfield  {journal} {\bibinfo
  {journal} {New J. Phys.}\ }\textbf {\bibinfo {volume} {19}},\ \bibinfo
  {pages} {113041} (\bibinfo {year} {2017})}\BibitemShut {NoStop}%
\bibitem [{\citenamefont {Svetlichny}(1987)}]{svetlichny87}%
  \BibitemOpen
  \bibfield  {author} {\bibinfo {author} {\bibfnamefont {G.}~\bibnamefont
  {Svetlichny}},\ }\bibfield  {title} {\enquote {\bibinfo {title}
  {Distinguishing three-body from two-body nonseparability by a {B}ell-type
  inequality},}\ }\href {\doibase 10.1103/PhysRevD.35.3066} {\bibfield
  {journal} {\bibinfo  {journal} {Phys. Rev. D}\ }\textbf {\bibinfo {volume}
  {35}},\ \bibinfo {pages} {3066} (\bibinfo {year} {1987})}\BibitemShut
  {NoStop}%
\bibitem [{\citenamefont {Seevinck}\ and\ \citenamefont
  {Svetlichny}(2002)}]{seevinck02}%
  \BibitemOpen
  \bibfield  {author} {\bibinfo {author} {\bibfnamefont {M.}~\bibnamefont
  {Seevinck}}\ and\ \bibinfo {author} {\bibfnamefont {G.}~\bibnamefont
  {Svetlichny}},\ }\bibfield  {title} {\enquote {\bibinfo {title} {Bell-type
  inequalities for partial separability in {$N$}-particle systems and quantum
  mechanical violations},}\ }\href {\doibase 10.1103/PhysRevLett.89.060401}
  {\bibfield  {journal} {\bibinfo  {journal} {Phys. Rev. Lett.}\ }\textbf
  {\bibinfo {volume} {89}},\ \bibinfo {pages} {060401} (\bibinfo {year}
  {2002})}\BibitemShut {NoStop}%
\bibitem [{\citenamefont {Collins}\ \emph {et~al.}(2002)\citenamefont
  {Collins}, \citenamefont {Gisin}, \citenamefont {Popescu}, \citenamefont
  {Roberts},\ and\ \citenamefont {Scarani}}]{collins02}%
  \BibitemOpen
  \bibfield  {author} {\bibinfo {author} {\bibfnamefont {D.}~\bibnamefont
  {Collins}}, \bibinfo {author} {\bibfnamefont {N.}~\bibnamefont {Gisin}},
  \bibinfo {author} {\bibfnamefont {S.}~\bibnamefont {Popescu}}, \bibinfo
  {author} {\bibfnamefont {D.}~\bibnamefont {Roberts}}, \ and\ \bibinfo
  {author} {\bibfnamefont {V.}~\bibnamefont {Scarani}},\ }\bibfield  {title}
  {\enquote {\bibinfo {title} {Bell-type inequalities to detect true
  $\mathit{n}$-body nonseparability},}\ }\href {\doibase
  10.1103/PhysRevLett.88.170405} {\bibfield  {journal} {\bibinfo  {journal}
  {Phys. Rev. Lett.}\ }\textbf {\bibinfo {volume} {88}},\ \bibinfo {pages}
  {170405} (\bibinfo {year} {2002})}\BibitemShut {NoStop}%
\bibitem [{\citenamefont {Gallego}\ \emph {et~al.}(2012)\citenamefont
  {Gallego}, \citenamefont {W\"{u}rflinger}, \citenamefont {Ac{\'{i}}n},\ and\
  \citenamefont {Navascu\'{e}s}}]{gallego12}%
  \BibitemOpen
  \bibfield  {author} {\bibinfo {author} {\bibfnamefont {R.}~\bibnamefont
  {Gallego}}, \bibinfo {author} {\bibfnamefont {L.~E.}\ \bibnamefont
  {W\"{u}rflinger}}, \bibinfo {author} {\bibfnamefont {A.}~\bibnamefont
  {Ac{\'{i}}n}}, \ and\ \bibinfo {author} {\bibfnamefont {M.}~\bibnamefont
  {Navascu\'{e}s}},\ }\bibfield  {title} {\enquote {\bibinfo {title}
  {Operational framework for nonlocality},}\ }\href {\doibase
  10.1103/PhysRevLett.109.070401} {\bibfield  {journal} {\bibinfo  {journal}
  {Phys. Rev. Lett.}\ }\textbf {\bibinfo {volume} {109}},\ \bibinfo {pages}
  {070401} (\bibinfo {year} {2012})}\BibitemShut {NoStop}%
\bibitem [{\citenamefont {Bancal}\ \emph {et~al.}(2013)\citenamefont {Bancal},
  \citenamefont {Barrett}, \citenamefont {Gisin},\ and\ \citenamefont
  {Pironio}}]{bancal13}%
  \BibitemOpen
  \bibfield  {author} {\bibinfo {author} {\bibfnamefont {J.-D.}\ \bibnamefont
  {Bancal}}, \bibinfo {author} {\bibfnamefont {J.}~\bibnamefont {Barrett}},
  \bibinfo {author} {\bibfnamefont {N.}~\bibnamefont {Gisin}}, \ and\ \bibinfo
  {author} {\bibfnamefont {S.}~\bibnamefont {Pironio}},\ }\bibfield  {title}
  {\enquote {\bibinfo {title} {Definitions of multipartite nonlocality},}\
  }\href {\doibase 10.1103/PhysRevA.88.014102} {\bibfield  {journal} {\bibinfo
  {journal} {Phys. Rev. A}\ }\textbf {\bibinfo {volume} {88}},\ \bibinfo
  {pages} {014102} (\bibinfo {year} {2013})}\BibitemShut {NoStop}%
\bibitem [{SM()}]{SM}%
  \BibitemOpen
  \href@noop {} {}\bibinfo {howpublished} {See Supplementary Material in the
  \href{https://arxiv.org/src/1708.07663v2/anc/Supplementary_Material.cdf}{arXiv `ancillary files'} for the full list of 2-causal inequalities in the
  tripartite lazy scenario and further analysis.}\BibitemShut {Stop}%
\end{thebibliography}%


%

\appendix

\section{Proof of the generalised 2-causal inequalities and their bounds}\label{appndx:genIneq}

\subsection{Proof of inequality~\eqref{eq:genIneq1}}

To prove that Eq.~\eqref{eq:genIneq1} is a valid 2-causal inequality for all $N$, it suffices to show that it holds for all deterministic 2-causal correlations.
For a nonempty strict subset $\K$ of $\N$, let $P^\text{det}(\vec{a}|\vec{x})=P^\text{det}(\vec{a}_{\K}|\vec{x}_{\K})P^\text{det}_{\vec{x}_{\K},\vec{a}_{\K}}(\vec{a}_{\N\backslash\K}|\vec{x}_{\N\backslash\K})$
be an arbitrary deterministic correlation compatible with the causal order $\K\prec \N\backslash\K$.
Then, since $P^\text{det}(\vec{a}_{\K}|\vec{x})=P^\text{det}(\vec{a}_{\K}|\vec{x}_{\K})$, it follows that
\begin{align}
	& P^\text{det}\big(\vec{a}_{\K}=\vec{1}\,|\,\vec{x}_{\K}=\vec{1},\vec{x}_{\N\backslash \K}=\vec{0}\big) \notag \\
	& \quad = P^\text{det}\big(\vec{a}_{\K}=\vec{1}\,|\,\vec{x}=\vec{1}\big) \ \ge \ P^\text{det}\big(\vec{a}=\vec{1}\,|\,\vec{x}=\vec{1}\big)
\end{align}
and hence $P^\text{det}\big(\vec{a}_{\K}=\vec{1}\,|\,\vec{x}_{\K}=\vec{1},\vec{x}_{\N\backslash \K}=\vec{0}\big) - P^\text{det}\big(\vec{a}=\vec{1}\,|\,\vec{x}=\vec{1}\big) \ge 0$.
Since $J_1(N)$ is then obtained by adding some more nonnegative terms $P^\text{det}\big(\vec{a}_{\K'}=\vec{1}\,|\,\cdots) \ge 0$, this proves the validity of Eq.~\eqref{eq:genIneq1} for any 2-causal correlation.

\subsection{Proof of inequalities~\eqref{eq:genIneq2}, \eqref{eq:McausalIneq} and~\eqref{sizeSbound} for $M$-causal and size-$S$-causal correlations}

The $M$-causal inequality~\eqref{eq:McausalIneq} and the size-$S$-causal inequality~\eqref{sizeSbound} are defined as different bounds for the expression $J_2(N)=\sum_{\{i,j\}\subset \N} L_N(i,j)$, with the summands defined in Eq.~\eqref{bilgynies}, while the 2-causal inequality~\eqref{eq:genIneq2} coincides with the particular cases $M=2$ and $S=N-1$. We shall first prove a bound for $J_2(N)$ that holds for $\Ptn$-causal correlations, for any partition $\Ptn$, and then use this bound to derive the corresponding $M$-causal, size-$S$-causal (and consequently the 2-causal) bounds.

Firstly, let us note that the observation made at the end of Sec.~\ref{subsec:Pcausal} that the response function determining the outputs of a deterministic $\Ptn$-causal correlation can be seen as processing deterministically one input after another and consequently defining a (dynamical) causal order between the subsets in $\Ptn$, also implies the following result (which will be used below and in the subsequent appendices):

\begin{proposition}\label{fixtwo}
For a deterministic $\Ptn$-causal correlation $P$, given two subsets $\A_\ell$ and $\A_m$ in $\Ptn$, the vector of inputs $\vec x_{\N \backslash (\A_\ell \cup \A_m)}$ for the parties that are neither in $\A_\ell$ nor in $\A_m$ determines a (not necessarily unique) causal order between $\A_\ell$ and $\A_m$, $\A_\ell \prec \A_m$ or $\A_m \prec \A_\ell$.

More technically: for any $\vec x_{\backslash \ell m} \coloneqq \vec x_{\N \backslash (\A_\ell \cup \A_m)}$, a deterministic $\Ptn$-causal correlation $P$ satisfies either $P(\vec a_{\A_\ell}|\vec x_{\backslash \ell m},\vec x_{A_\ell},\vec x_{\A_m}) = P(\vec a_{\A_\ell}|\vec x_{\backslash \ell m},\vec x_{A_\ell},\vec x_{\A_m}^{\,\prime})$ for all $\vec x_{A_\ell}, \vec x_{\A_m}, \vec x_{\A_m}^{\,\prime}, \vec a_{\A_\ell}$, or $P(\vec a_{\A_m}|\vec x_{\backslash \ell m},\vec x_{A_\ell},\vec x_{\A_m}) = P(\vec a_{\A_m}|\vec x_{\backslash \ell m},\vec x_{A_\ell}^{\,\prime},\vec x_{\A_m})$ for all $\vec x_{A_m}, \vec x_{A_\ell}^{\,\prime},\vec x_{\A_m},\vec a_{\A_m}$---i.e., in short, either $P(\vec a_{\A_\ell}|\vec x) = P(\vec a_{\A_\ell}|\vec x_{\N \backslash \A_m})$ or $P(\vec a_{\A_m}|\vec x) = P(\vec a_{\A_m}|\vec x_{\N \backslash \A_\ell})$.
\end{proposition} 

To derive a $\Ptn$-causal bound for $J_2(N)$ for a given partition $\Ptn=\{\A_1,\dots,\A_{|\Ptn|}\}$, it is sufficient to find a bound that holds for any deterministic $\Ptn$-causal correlation $P$. We will bound $J_2(N)$ by bounding each individual term $L_N(i,j)$ in the sum.
There are two cases to be considered: whether \textit{i)} the parties $A_i$ and $A_j$ are in different subsets of $\Ptn$, i.e.\ $i\in\A_\ell$, $j\in \A_m$ with $\ell \neq m$; or \textit{ii)} both parties are in the same subset: $i,j\in \A_\ell$.

\begin{enumerate}

\item[\textit{i)}] According to Proposition~\ref{fixtwo}, the inputs $\vec x_{\N\backslash (\A_\ell \cup \A_m)} = \vec 0$ imply either the order $\A_\ell \prec \A_m$, or $\A_m \prec \A_\ell$ for $P$. In the first case,
\begin{align}
	 & \quad P\big(a_i=1|x_ix_j=10,\vec{x}_{\N\backslash \{i,j\}}=\vec{0}\big) \notag \\
	 & \qquad = P\big(a_i=1|x_ix_j=11,\vec{x}_{\N\backslash \{i,j\}}=\vec{0}\big),
\end{align}
which implies
\begin{align}
	 &  P\big(a_i=1|x_ix_j=10,\vec{x}_{\N\backslash \{i,j\}}=\vec{0}\big) \notag \\
	 & \quad - P\big(a_ia_j=11|x_ix_j=11,\vec{x}_{\N\backslash \{i,j\}}=\vec{0}\big) \ge 0,
\end{align}
and therefore (after adding a nonnegative term) $L_N(i,j)\ge 0$.
An analogous argument shows that $L_N(i,j)\ge 0$ also in the case that one has $\A_m \prec \A_\ell$ for $P$ when $\vec x_{\N\backslash (\A_\ell \cup \A_m)} = \vec 0$.

\item[\textit{ii)}] If the parties $A_i$ and $A_j$ belong to the same subset $\A_\ell$, they can share arbitrary correlations and thus win the (conditional) \mbox{LGYNI} game perfectly.
In that case we have $L_N(i,j)\ge -1$, which is the minimum algebraic bound.

\end{enumerate}

Combining the two cases, we thus have, for any $\Ptn$-causal correlation,
\begin{align}
	J_2(N) & =\sum_{\{i,j\}\subset \N} L_N(i,j) \notag \\[-1mm]
	& \ge \hspace{-3mm} \sum_{\substack{\{i,j\}\in \A_\ell \\ \textrm{for some } \A_\ell \in \Ptn }} \!\!\!\!\!\! (-1) \notag\\ &= \ - \sum_{\A_\ell \in \Ptn} \binom{|\A_\ell|}{2} \coloneqq L(\Ptn).  \label{eq:J2N_LP}
\end{align}

\medskip

In order to prove the $M$-causal bound~\eqref{eq:McausalIneq}, we shall now prove that among all partitions $\Ptn$ containing a fixed number $\mathfrak{m}$ of subsets, the quantity $L(\Ptn)$ defined above is minimal when $\Ptn$ consists of $\mathfrak{m}-1$ singletons, and one subset containing the remaining $N-\mathfrak{m}+1$ parties.
Assume for the sake of contradiction that this is not the case, so that the minimum is obtained for a partition $\Ptn$ that contains at least two subsets $\A_\ell$ and $\A_m$ that are not singletons, for  which we assume $|\A_\ell| \ge |\A_m|\, (\ge 2)$. 
Let then $k \in \A_m$, and define the partition $\Ptn'$ obtained from $\Ptn$ by replacing $\A_\ell$ and $\A_m$ by $\A_\ell' = \A_\ell \cup \{k\}$ and $\A_m' = \A_m \backslash\{k\}$, respectively (note that the assumption that $|\A_m| \ge 2$ ensures that $\A_m'$ remains nonempty). 
One then has
\begin{align}\label{eq:Ldiff1}
	& L(\Ptn')-L(\Ptn) \notag \\
	& \quad = - \Big[ \binom{|\A_\ell'|}{2} + \binom{|\A_m'|}{2} \Big] + \Big[ \binom{|\A_\ell|}{2} + \binom{|\A_m|}{2} \Big] \notag \\
	& \quad = - |\A_\ell| + |\A_m| - 1 \ < \ 0,
\end{align}
in contradiction with the assumption that $\Ptn$ minimised $L$. 
For a given $N$ it then follows that
\begin{equation}
	\min_{\Ptn:|\Ptn|=\mathfrak{m}}L(\Ptn)=-\binom{N-\mathfrak{m}+1}{2},
\end{equation}
and therefore, from Eq.~\eqref{eq:J2N_LP},
\begin{equation}
	J_2(N) \ge -\binom{N-|\Ptn|+1}{2}.
\end{equation}
Finally, we note that $|\Ptn| \ge M$ implies $-\binom{N-|\Ptn|+1}{2} \ge -\binom{N-M+1}{2}$, which concludes the proof that Eq.~\eqref{eq:McausalIneq} holds for all $M$-causal correlations.

\medskip

In order now to prove the bound~\eqref{sizeSbound} for size-$S$-causal correlations, we show that among all partitions $\Ptn$ with $\mathfrak{s}(\Ptn)\le S$, $L(\Ptn)$ from Eq.~\eqref{eq:J2N_LP} is minimised for the partition containing $\left\lfloor \frac{N}{S}\right\rfloor$ groups of $S$ parties, and (if $N$ is not a multiple of $S$) a last group with the remaining $N - \left\lfloor \frac{N}{S}\right\rfloor S$ parties---for which $L(\Ptn)$ is indeed equal to the right-hand side of Eq.~\eqref{sizeSbound}. 
Assume again for the sake of contradiction that this is not the case, so that the minimum is obtained for a partition $\Ptn$ containing at least two subsets $\A_\ell$ and $\A_m$ of less than $S$ parties, for which we take $|\A_m|\le|\A_\ell| < S$.
If $|\A_m|>1$, one can follow the same reasoning as in the proof of the $M$-causal bound above: take $k\in\A_m$ and consider the partition $\Ptn'$ obtained by replacing $\A_\ell$ and $\A_m$ by $\A_\ell'=\A_\ell\cup\{k\}$ and $\A_m'=\A_m\setminus\{k\}$, respectively.
Note that since we assumed $|\A_\ell|<S$, we have $|\A_\ell'|\le S$ and $\mathfrak{s}(\Ptn')\le S$.
Eq.~\eqref{eq:Ldiff1} then holds again, in contradiction with the assumption that $\Ptn$ minimised $L$.
In the case when $|\A_m|=1$, consider instead the partition $\Ptn'$ formed by merging $\A_\ell$ and $\A_m$ into a new subset $\A_\ell'=\A_\ell\cup\A_m$ (so that $|\A_\ell'|=|\A_\ell|+1 \le S$ and we still have $\mathfrak{s}(\Ptn')\le S$).
We then have
\begin{align}\label{eq:Ldiff2}
	L(\Ptn')-L(\Ptn) &= - \binom{|\A_\ell'|}{2} +  \binom{|\A_\ell|}{2} \notag\\
	 &= - |\A_\ell| \ < \ 0,
\end{align}
again in contradiction with the assumption that $\Ptn$ minimised $L$, which concludes the proof that Eq.~\eqref{sizeSbound} holds for all size-$S$-causal correlations.

\section{Separation of $\Ptn$-causal polytopes} \label{Pseparationproof}

In this appendix we shall prove Theorem~\ref{Pseparation}, which states that there are no nontrivial inclusions among $\Ptn$-causal polytopes. Before that, we start by introducing a useful family of deterministic $\Ptn$-causal correlations.

\subsection{A family of deterministic $\Ptn$-causal correlations}
\label{app_subsec:P_P_sigma}

The $N$-partite correlations $P_{\Ptn,\sigma}^\text{det}$ we introduce here are defined for a given partition $\Ptn = \{ \A_1, \ldots, \A_{|\Ptn|} \}$ of $\N$ and a given permutation $\sigma$ of its $|\Ptn|$ elements.

We consider the lazy scenario, where each party has a binary input $x_k = 0,1$ with a fixed output $a_k = 0$ for $x_k = 0$, and a binary output $a_k = 0,1$ for $x_k = 1$.
For each subset $\A \in \Ptn$ and a vector of inputs $\vec x$, we define the bit
\begin{equation} 
z_\A \coloneqq \prod_{k\in \A} x_k.
\end{equation}
We then define the deterministic response function $\vec\alpha_{\Ptn,\sigma}$ such that, for each party $A_k$ belonging to a subset $\A_\ell$ of $\Ptn$, we have
\begin{equation} \label{eq:def_resp_fcn_Ppsigma}
\big(\vec\alpha_{\Ptn,\sigma}(\vec{x})\big)_k \ \coloneqq \!\! \prod_{m:\,\sigma(m) \leq \sigma(\ell)} \!\!\! \! z_{\A_m}.
\end{equation}
The correlation $P_{\Ptn,\sigma}^\text{det}$ is then defined as
\begin{equation} \label{eq:defPpsigma}
P_{\Ptn,\sigma}^\text{det}(\vec a|\vec x) \ \coloneqq \delta_{\vec a,\vec\alpha_{\Ptn,\sigma}(\vec{x})}.
\end{equation}

In other words, each party $A_k$ in some subset $\A_\ell \in \Ptn$ outputs the product of the inputs of all parties that came before itself according to the partition $\Ptn$ and the causal order $\A_{\sigma(1)} \prec \A_{\sigma(2)} \prec \cdots \prec \A_{\sigma(|\Ptn|)}$ defined by the permutation $\sigma$, including all parties in the same subset $\A_\ell$. 
Clearly the correlation $P_{\Ptn,\sigma}^\text{det}$ is compatible with this fixed causal order, and is therefore $\Ptn$-causal; as it is deterministic, it corresponds to a vertex of the $\Ptn$-causal polytope.

Note that each party outputs $a_k=0$ whenever $x_k=0$, as required in the lazy scenario. 
The correlations $P_{\Ptn,\sigma}^\text{det}$ can also straightforwardly be generalised to more complex scenarios with more inputs and outputs, by simply never outputting the other possible outputs, and, e.g., always outputting $0$ for any other possible input.
Hence, the proofs below, which use $P_{\Ptn,\sigma}^\text{det}$ as an explicit example, apply to any scenario where each party has at least two possible inputs, and at least two possible outputs for one of their inputs.

\subsection{Proof of Theorem~\ref{Pseparation}}
\label{app_subsec:proof_Pseparation}

Coming back to the theorem, we shall prove that the $\Ptn$-causal polytope is not contained in the $\Ptn'$-causal one by exhibiting a $\Ptn$-causal correlation (from the family introduced above) that is not $\Ptn'$-causal. 
The proof that the $\Ptn'$-causal polytope is not contained in the $\Ptn$-causal one then follows by symmetry. 
We distinguish two cases, whether \textit{i)} $\Ptn'$ is a coarse-graining of $\Ptn$ or \textit{ii)} $\Ptn'$ is not a coarse-graining of $\Ptn$.

\begin{enumerate}

\item[\textit{i)}] If $\Ptn'$ (with $|\Ptn'|>1$) is a coarse-graining of $\Ptn$, then one can find two subsets $\A_\ell$ and $\A_{\ell'}$ in $\Ptn$ that are grouped together in some subset $\A_{\ell \ell'}'$ in $\Ptn'$, and a third subset $\A_m$ in $\Ptn$ that is contained in a different subset $\A_m'$ of $\Ptn'$.
Let $\sigma$ be a permutation of $\Ptn$ which defines a causal order between its elements such that $\A_\ell \prec \A_m \prec \A_{\ell'}$.
The correlation $P_{\Ptn,\sigma}^\text{det}$ as defined in Eq.~\eqref{eq:defPpsigma} is then $\Ptn$-causal, but not $\Ptn'$-causal.
Intuitively, this is because we cannot order $\A_{\ell \ell'}'$ (in which $\A_\ell$ and $\A_{\ell'}$ are grouped together) against $\A_m'$ (which contains $\A_m$).
More specifically, for $\vec x_{\N \backslash (\A_\ell \cup \A_{\ell'} \cup \A_m)} = \vec 1$ (so that in particular, $\vec x_{\N \backslash (\A_{\ell \ell'}' \cup \A_m')} = \vec 1$), the response function $\vec\alpha_{\Ptn,\sigma}$ defined in Eq.~\eqref{eq:def_resp_fcn_Ppsigma} gives $a_k=\big(\vec\alpha_{\Ptn,\sigma}(\vec x)\big)_k = z_{\A_\ell}z_{\A_m}z_{\A_{\ell'}}$ if $k \in \A_{\ell'}$ and $a_k=\big(\vec\alpha_{\Ptn,\sigma}(\vec x)\big)_k = z_{\A_\ell}z_{\A_m}$ if $k \in \A_m$. 
Hence, $P(\vec a_{\A_{\ell \ell'}'}|\vec x)$ depends nontrivially on $\vec x_{\A_m'}$ (via $z_{\A_m}$) while $P(\vec a_{\A_m'}|\vec x)$ depends nontrivially on $\vec x_{\A_{\ell \ell'}'}$ (via $z_{\A_\ell}$). 
According to Proposition~\ref{fixtwo}, this implies that $P_{\Ptn,\sigma}^\text{det}$ indeed cannot be $\Ptn'$-causal.

\item[\textit{ii)}] If $\Ptn'$ is not a coarse-graining of $\Ptn$, then one can find two parties $A_i, A_j$ that belong to the same subset $\A_{ij}$ of $\Ptn$, but belong to two distinct subsets of $\Ptn'$, i.e.\ $A_i \in \A_i'$, $A_j \in \A_j'$.
Let now $\sigma$ be \textit{any} permutation of $\Ptn$. The correlation $P_{\Ptn,\sigma}^\text{det}$ as defined in Eq.~\eqref{eq:defPpsigma} is then $\Ptn$-causal, but not $\Ptn'$-causal.
Intuitively, this is because the parties $A_i$ and $A_j$ cannot be separated in the definition of $P_{\Ptn,\sigma}^\text{det}$.
More specifically, for $\vec x_{\N \backslash \{i,j\}} = \vec 1$ (so that in particular, $\vec x_{\N \backslash (\A_i' \cup \A_j')} = \vec 1$), the response function $\vec\alpha_{\Ptn,\sigma}$ gives $a_k=\big(\vec\alpha_{\Ptn,\sigma}(\vec x)\big)_k = z_{\A_{ij}} =x_i x_j$ for both $k=i$ and $k=j$. 
Hence, $P(\vec a_{\A_i'}|\vec x)$ depends nontrivially on $\vec x_{\A_j'}$ (via $x_j$) while $P(\vec a_{\A_j'}|\vec x)$ depends nontrivially on $\vec x_{\A_i'}$ (via $x_i$). 
According to Proposition~\ref{fixtwo}, this implies that $P_{\Ptn,\sigma}^\text{det}$ indeed cannot be $\Ptn'$-causal.

\end{enumerate}

\section{A 4-partite fully causal correlation with dynamical order that is not a convex mixture of $\Ptn_j'$-causal correlations with $|\Ptn_j'| = 3$.} \label{dynamicalcounter}

We provide here an explicit counterexample to the question raised at the end of Sec.~\ref{subsec:nonincl_Pcausal}, of whether a $\Ptn$-causal correlation can always be written as a convex combination of $\Ptn_j'$-causal correlations for various partitions $\Ptn_j'$ with a fixed number of subsets $|\Ptn_j'| = \mathfrak{m}' < |\Ptn|$.

\begin{figure}
	\begin{center}
		\includegraphics[width=.7\columnwidth]{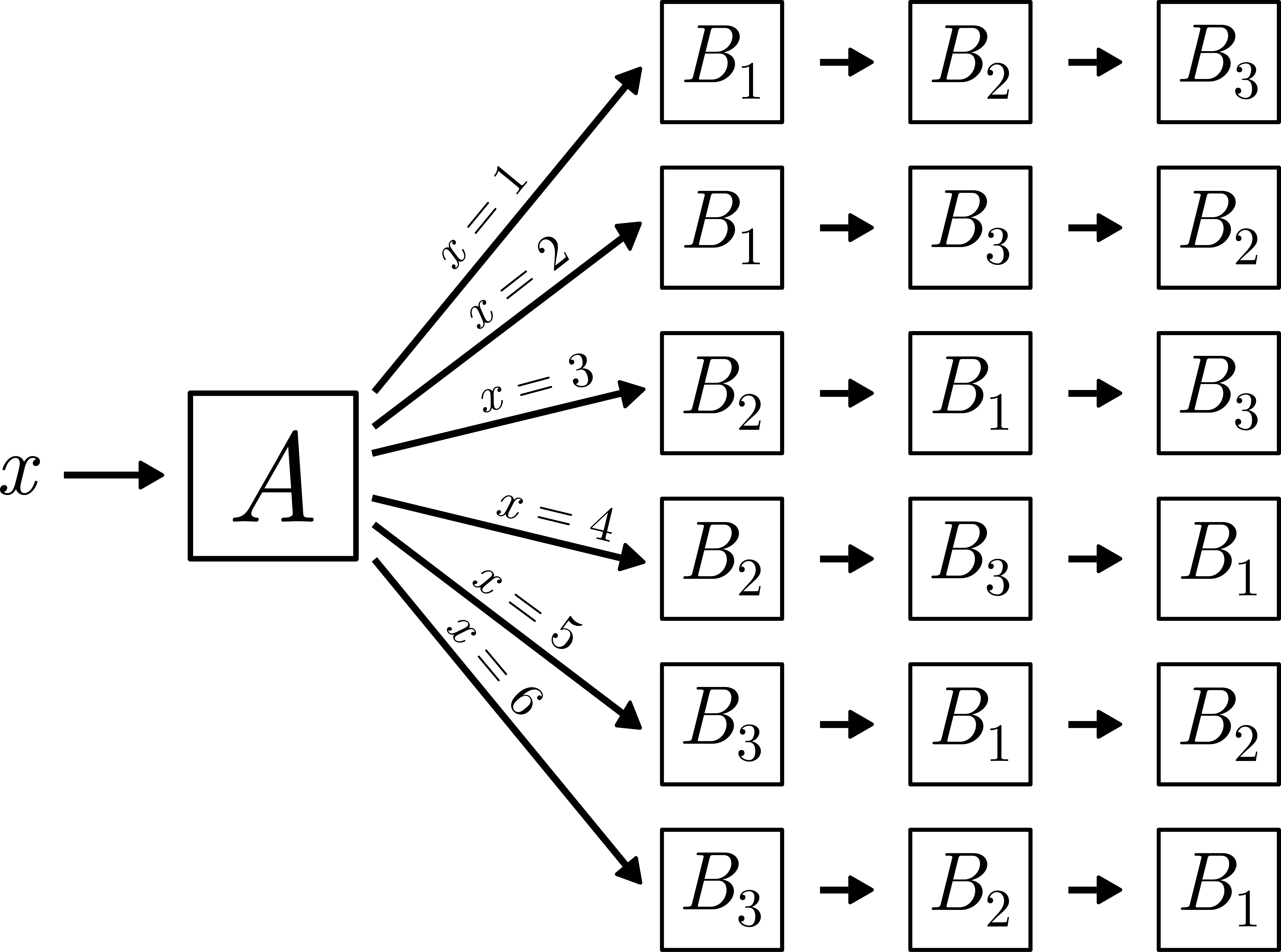}
	\end{center}
\caption{The causal structure sketched above provides an example of a 4-partite fully causal correlation that is not a convex mixture of $\Ptn_j'$-causal correlations with $|\Ptn_j'| = 3$ (see text for details).}
\label{fig:4causal}
\end{figure}

As we noted, such a counterexample requires $\mathfrak{m}' \ge 3$ (and hence $|\Ptn| \ge 4$), as well as a dynamical causal order.
Consider thus a 4-partite scenario, with parties $A$, $B_1$, $B_2$ and $B_3$. 
Party $A$ receives as input a 6-valued variable $x$ (and has no output); 
$A$'s input determines the causal order of the three subsequent parties $B_k$ (see Fig.~\ref{fig:4causal}), with each possible value of $x$ corresponding to one of the six possible permutations, denoted by $\sigma_x$. For parties $B_k$ we consider the lazy scenario, with inputs $y_k \in\{0,1\}$ and outputs $b_k = 0$  if $y_k=0$, $b_k \in\{0,1\}$ if $y_k=1$. 
We then define the deterministic correlation $P^{\textrm{det}}$ by the response functions
\begin{equation}
b_{\sigma_x(1)} \!=\! 0, \ b_{\sigma_x(2)} \!=\! y_{\sigma_x(1)}y_{\sigma_x(2)}, \ b_{\sigma_x(3)} \!=\! y_{\sigma_x(2)}y_{\sigma_x(3)}.
\end{equation}

While the correlation $P^{\textrm{det}}$ thus obtained is fully causal (i.e., it is $\Ptn$-causal for the ``full partition'' $\Ptn$ such that $|\Ptn|=N=4$), it is not $\Ptn'$-causal for any 3-subset partition $\Ptn'$ of $\{A,B_1,B_2,B_3\}$---which also implies, since $P^{\textrm{det}}$ is deterministic, that it is not decomposable as a convex combination of $\Ptn'$-causal correlations for various 3-subset partitions $\Ptn'$ either.
Indeed, such a $\Ptn'$ would contain (2 singletons and) a pair of parties grouped together, $\{A,B_i\}$ or $\{B_i,B_j\}$. 
Consider the first case first: as $P^{\textrm{det}}$ is deterministic, and the outputs of all parties $B_k$ depend on $x$, any $\Ptn'$-causal correlation should be compatible with the subset $\{A,B_i\}$ being first, with therefore $b_i$ independent of $y_k$ for $k\neq i$; this, however, cannot be because, for every $i=1,2,3$, we can find $x$ such that $i=\sigma_x(2)$, so that $b_i= y_{\sigma_x(1)} y_i$, which depends on $y_k$ with $k = {\sigma_x(1)} \neq i$. 
In the second case where $\Ptn' = \{\{A\},\{B_i,B_j\},\{B_k\}\}$, according to Proposition~\ref{fixtwo}, a deterministic $\Ptn'$-causal correlation must be such that for each given value of $x$ one must either have that $b_i$ and $b_j$ should be independent of $y_k$, or that $b_k$ is independent of $y_i$ and $y_j$; this is however not satisfied for the value of $x$ such that $\sigma_x(1)=i$, $\sigma_x(2)=k$, $\sigma_x(3)=j$.

In short, for any pair of parties there exists some input $x$ of party $A$ for which a third party must act in between the said pair, so that this pair of parties cannot be causally ordered with the other two (singletons of) parties. 
This shows that the correlation $P^{\textrm{det}}$ defined above is indeed not $\Ptn'$-causal for any 3-subset partition $\Ptn'$---and as said above, being deterministic, it is not a convex mixture of $\Ptn'$-causal partitions for various such partitions $\Ptn'$ either.

\section{Proof of completeness of Theorem~\ref{thm:inclusions}}\label{app_sec:inclusions}

In order to prove that Theorem~\ref{thm:inclusions} completely characterises the possible inclusions between $M$-causal and size-$S$-causal polytopes, we first prove the following lemma regarding non-inclusions between $\Ptn$-causal polytopes (which is perhaps of interest in and of itself).
\begin{lemma}\label{noinclusions}
Given a partition $\Ptn$ and a set of partitions $\left\{\Ptn_1',\dots,\Ptn_r'\right\}$, none of which is a coarse-graining of $\Ptn$, the convex hull of the $\Ptn_j'$-causal polytopes, $j=1,\dots,r$, does not contain the $\Ptn$-causal polytope.
\end{lemma}
\begin{proof}
	It suffices here to show that, if no partition among $\Ptn_1',\dots,\Ptn_r'$ is a coarse-graining of a partition $\Ptn$, it is possible to find a deterministic $\Ptn$-causal correlation that is not $\Ptn_j'$-causal for any $j=1,\dots,r$. 
	The given correlation being deterministic, this will indeed imply that it is also not a convex mixture of $\Ptn_j'$-causal correlations.

	We can again take the correlation $P_{\Ptn,\sigma}^\text{det}$ defined in Eq.~\eqref{eq:defPpsigma}, for any choice of the permutation $\sigma$. 
	Recall that for this correlation the output of each party depends nontrivially on the inputs of all parties in the same subset. 
	As already established for case \emph{ii)} in Appendix~\ref{app_subsec:proof_Pseparation}, no such correlation is $\Ptn_j'$-causal for any partition $\Ptn_j'$ that is not a coarse-graining of $\Ptn$, which proves the result.
\end{proof}
	Note that the assumption that none of the partitions $\Ptn_j'$ is a coarse-graining of $\Ptn$ is crucial in the above proof, and the conclusion of the theorem does not necessarily hold otherwise: as noted in Sec.~\ref{subsec:nonincl_Pcausal} already, Eq.~\eqref{eqdef:P_causal} indeed shows that a $\Ptn$-causal correlation, with $\Ptn = \{\A_\ell\}_\ell$, can be written as a convex combination of $\Ptn_\ell'$-causal correlations, with the partitions $\Ptn_\ell' = \{\A_\ell, \N \backslash \A_\ell\}$ being coarse-grainings of $\Ptn$.

\medskip

Let us now prove the completeness of Theorem~\ref{thm:inclusions}.
To this end, let us consider first a partition $\Ptn$ with $|\Ptn| = M$ that consists of $M-1$ singletons and an $(N{-}M{+}1)$-partite subset. 
Such a partition saturates the first inequality in Eq.~\eqref{eq:ineqs_S_M}, i.e., it satisfies $\mathfrak{s}(\Ptn) = N-M+1$. 
Let us then take $S < N-M+1$. 
The size-$S$-causal polytope is, by definition, the convex hull of all $\Ptn_j'$-causal polytopes for all partitions $\Ptn_j'$ with $\mathfrak{s}(\Ptn_j') \le S$. 
None of these partitions can be a coarse-graining of $\Ptn$, as this would imply (since coarse-graining can only increase the size of the largest subset in a partition) $\mathfrak{s}(\Ptn_j') \ge \mathfrak{s}(\Ptn) = N-M+1 > S$, in contradiction with $\mathfrak{s}(\Ptn_j') \le S$. 
But then Lemma~\ref{noinclusions} shows that the $\Ptn$-causal polytope is not contained in the size-$S$-causal polytope, and (since $|\Ptn| = M$) this thus implies that the $M$-causal polytope is not contained in the size-$S$-causal polytope. 

Similarly, consider a partition $\Ptn$ with $\mathfrak{s}(\Ptn) = S$, that consists of $\left\lfloor \frac{N}{S}\right\rfloor$ groups of $S$ parties and, if $N$ is not a multiple of $S$, a final group containing the remaining $N {-} \left\lfloor \frac{N}{S}\right\rfloor S$ parties. 
Such a partition thus contains $|\Ptn| = \left\lceil \frac{N}{S}\right\rceil$ subsets. 
Let us now take $M > \left\lceil \frac{N}{S}\right\rceil$. 
The $M$-causal polytope is, again by definition, the convex hull of all $\Ptn_j'$-causal polytopes for all partitions $\Ptn_j'$ with $|\Ptn_j'| \ge M$. 
None of these partitions can be a coarse-graining of $\Ptn$, as this would imply (since coarse-graining can only decrease the number of subsets in a partition) $|\Ptn_j'| \le |\Ptn| = \left\lceil \frac{N}{S}\right\rceil < M$, in contradiction with $|\Ptn_j'| \ge M$. 
Lemma~\ref{noinclusions} then again shows that the $\Ptn$-causal polytope is not contained in the $M$-causal polytope, and (since $\mathfrak{s}(\Ptn) = S$) this then implies that the size-$S$-causal polytope is not contained in the $M$-causal polytope, which completes the proof.

Finally, let us also note that since no partition $\Ptn'$ with $\left|\Ptn'\right| \ge M' > M$ is a coarse-graining of any partition $\Ptn$ with $|\Ptn| = M$, and since no partition $\Ptn'$ with $\mathfrak{s}(\Ptn') \le S' < S$ is a coarse-graining of any partition $\Ptn$ with $\mathfrak{s}(\Ptn) = S$, invoking Lemma~\ref{noinclusions} also provides a proof (as an alternative to our use of the families of $M$-causal and size-$S$-causal inequalities~\eqref{eq:McausalIneq} and~\eqref{sizeSbound} before) that all inclusions among $M$-causal and among size-$S$-causal polytopes are strict.

\end{document}